\documentclass{article}

\usepackage{styles/arxiv}

\usepackage{amsmath,amsfonts,amssymb,amsthm,booktabs,color,epsfig,graphicx,hyperref,url}
\usepackage{natbib}
\usepackage{mathtools}

\usepackage{multirow}
\usepackage{subcaption}
\usepackage{bbm}
\usepackage{todonotes} 
\usepackage{diagbox} 

\theoremstyle{plain}
\newtheorem{theorem}{Theorem}

\newtheorem{proposition}{Proposition}

\theoremstyle{definition}

\newtheorem{assumption}{Assumption}

\newcommand{\atan}{\text{atan}}
\newcommand{\zerob}{\pmb 0}
\newcommand{\R}{\mathbb{R}}

\newcommand{\Z}{\mathbb{Z}}

\newcommand{\E}{\mathbb{E}}
\newcommand{\Prob}{\mathbb{P}}
\newcommand{\var}{\text{Var}}
\newcommand{\cov}{\text{Cov}}
\newcommand{\one}{\mathbbm{1}}

\newcommand{\mub}{\boldsymbol{\mu}}

\newcommand{\lambdab}{\boldsymbol{\lambda}}
\newcommand{\thetab}{\boldsymbol{\theta}}
\newcommand{\Thetab}{\boldsymbol{\Theta}}
\newcommand{\taub}{\boldsymbol{\tau}}
\newcommand{\varthetab}{\boldsymbol{\vartheta}}
\newcommand{\xb}{\boldsymbol{x}}
\newcommand{\alphab}{\boldsymbol{\alpha}}

\newcommand{\deltab}{\boldsymbol{\delta}}
\newcommand{\Rb}{\boldsymbol{R}}
\newcommand{\kappab}{\boldsymbol{\kappa}}
\newcommand{\xib}{\boldsymbol{\xi}}
\newcommand{\Xb}{\boldsymbol{X}}
\newcommand{\Vb}{\boldsymbol{V}}

\newcommand{\pb}{\boldsymbol{p}}
\newcommand{\eb}{\boldsymbol{e}}

\title{Arctan-skewing: a new skewing scheme on the hypertorus}

\author{
	\href{https://orcid.org/0000-0002-2290-8437}{\includegraphics[scale=0.06]{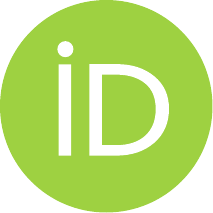}\hspace{1mm}Christophe Ley} \\
    Department of Mathematics \\
	University of Luxembourg\\
	\texttt{christophe.ley@uni.lu} \\
    \And
	\href{https://orcid.org/0009-0006-1372-7032}{\includegraphics[scale=0.06]{styles/orcid.pdf}\hspace{1mm}Sophia Loizidou} \\
    Department of Mathematics \\
	University of Luxembourg\\
	\texttt{sophia.loizidou@uni.lu} \\
}

\begin{document}

\maketitle

\begin{abstract}
In this paper we propose and study arctan-skewing, a general approach to skew any symmetric distribution on the unit hypertorus. Arctan-skewing retains the main advantages of the popular sine-skewing mechanism, but improves on it by overcoming its two major limitations. These are a  condition on the skewness parameter and a singular Fisher Information matrix in the vicinity of symmetry when combined with well-known symmetric models from the literature such as the circular von Mises or bivariate Cosine models. Arctan-skewing thus represents a compelling alternative to model skewed data on the circle and torus, notably from the field of protein bioinformatics. After a thorough investigation of its stochastic properties, we present a maximum likelihood estimation procedure and devise optimal tests for symmetry against our new class of skewed alternatives. Finally we show that arctan-skewed models better fit protein data compared to their sine-skewed counterparts.
\end{abstract}

\keywords{Directional statistics;
Flexible modeling;
Skewness;
Structural bioinformatics;
Tests for symmetry;
Toroidal data.}

\section{Introduction} \label{sec: intro}

Probability distributions are the building blocks of statistical modelling and inference. 
Flexible modelling is an important part of statistical research, as it is concerned with the construction of new probability distributions that are versatile enough to model complex data exhibiting features such as asymmetry, heavy/light tails or dependence structures, features that the  Gaussian distribution fails to represent in $\R^k, k\geq1$. 
A very large body of literature considers the problem of modelling asymmetric data given their omnipresence in fields like finance, biology, environmental sciences or meteorology, to mention but these. 
Besides several ad hoc proposals, general approaches have been presented over the past decades, including finite mixture models, mean-variance mixtures, two-piece models, the transformation approach,  the  probability integral transformations,  power transformations, or copulas; see \cite{jones2015families}, \cite{babic2019comparison} and \cite{ley2021flexible} for an overview. 
The arguably most popular proposal is the so-called symmetry modulation or Azzalini-skewing \citep{azzalini2018symmetry}, where a symmetric probability density function (pdf) $f(\cdot)$ is multiplied by a so-called skewing function $\Pi(\cdot;\deltab)$ satisfying conditions similar to those that cumulative distribution functions (cdfs) typically fulfill and where $\deltab$ is a skewing parameter.
The resulting density $f(\cdot)\Pi(\cdot;\deltab)$  is asymmetric unless $\deltab=\zerob$, in which case one retrieves the original density $f(\cdot)$. 
A highly appealing property of this symmetry modulation lies in the fact that no normalizing constant needs to be calculated, as $f(\cdot)\Pi(\cdot;\deltab)$ directly integrates to 1 under these conditions. 
Famous representatives of symmetry modulation are the scalar skew-normal distribution of \cite{azzalini_class_1985}, the multivariate skew-normal of \cite{azzalini1996multivariate}, the skew-elliptical distributions of \cite{azzalini1999statistical}, the skew-$t$ distribution of \cite{azzalini2003distributions}, and the skew-symmetric family of \cite{wang2004skew}. Examples on other manifolds include the sine-skewing on the unit circle by \cite{abe2011sine}, the skew-rotationally-symmetric distributions on the unit sphere of \cite{ley_skew-rotationally-symmetric_2017}, the skewing on general manifolds by \cite{jupp_general_2016}, or the sine-skewing mechanism on the hypertorus by \cite{ameijeiras-alonso_sine-skewed_2022}.

Besides the ease of construction and good interpretability of symmetry modulated models, they furthermore enjoy appealing stochastic properties such as a simple random number generation mechanism. 
However, from an inferential viewpoint some of these models are infamous for suffering from a singular Fisher Information matrix in the vicinity of symmetry, meaning that the score vectors for location and skewness become co-linear when $\deltab = \zerob$.
This singularity prevents the construction of non-trivial tests for symmetry. 
More specifically, the one- and multi-dimensional skew-normal distribution is infamous for this issue, while the skew-$t$ is free from it. \cite{Hallin_singularity2012} and \cite{HallinLey2014} have solved the long-standing question of which Azzalini-type models suffer from singular Fisher Information matrices on the real line and in $\R^k$ by providing a complete characterization of singularity-inducing mismatches between $f(\cdot)$ and $\Pi(\cdot;\deltab)$. 

In this paper, we are interested in the hypertorus as support. Mathematically speaking, the $d$-dimensional hypertorus $\mathcal{T}_d = \left\{ \thetab \in [-\pi,\pi)^d \right\}$ is the Cartesian product of $d$ unit circles.
Important datasets from various domains are represented as points on the torus, such as wind direction measured at different times during the day \citep{johnson_measures_1977, kato_distribution_2009, shieh_inferences_2005}, directions of steepest descent before and after an earthquake \citep{rivest_decentred_1997},
and morphological data from human neurons \citep{puerta_regularized_2015}. 
A domain that has given rise
in the past two decades to a lot of such toroidal data is bioinformatics, where dihedral angles from proteins can be
viewed as data on the torus \citep{TWCM, boomsma_generative_2008, golden2017generative}. 
As mentioned before, symmetry modulation on the torus has been proposed by \cite{ameijeiras-alonso_sine-skewed_2022}, resulting in the so-called sine-skewed distributions. 
Starting from a symmetric density $\thetab\mapsto f_0(\thetab - \mub; \varthetab), \thetab\in\mathcal{T}_d,$ on the unit hypertorus of dimension $d$, with symmetry center (or location parameter) $\mub \in \mathcal{T}_d$ and $\varthetab$ the corresponding set of non-location parameters, sine-skewed densities are of the form  
\begin{align}\label{eq: sine_skewed_density}
    \thetab\mapsto f_{\mub, \lambdab}(\thetab; \varthetab) := f_0(\thetab - \mub; \varthetab)\left(1+\sum_{j=1}^{d} \lambda_j \sin (\theta_j-\mu_j)\right) \quad \mbox{subject to}\,\, 
    \sum_{j=1}^{d} |\lambda_j| \leq 1,
\end{align}
where  $\lambdab \in [-1,1]^d$ is the skewness parameter. 
Sine-skewed distributions on the torus play an important role in the modelling of protein structures and have been implemented by computer scientists in the deep probabilistic programming language Pyro. 
Prior structures for Bayesian inference for the case when $f_0$ is the Sine distribution of \cite{singh_probabilistic_2002} have been proposed in \cite{ronning_time_efficient_2021}, while efficient yet robust (in the sense of being valid under a broad class of distributions satisfying the null hypothesis) tests for symmetry, both about specified and unspecified center, have been proposed on the torus against  sine-skewed alternatives in \cite{symmetry_test}. 
However, as for their equivalents from Euclidean spaces $\R^k$, the sine-skewed distributions also suffer from singular Fisher Information matrices in the vicinity of symmetry. 
A complete characterization of which symmetric densities $f_0$ lead to singularity has been provided in \cite{ley_simple_2014} for $d=1$ and in \cite{schutte_singularity_2026} for any dimension, showcasing in particular that this occurs when $f_0$ is the circular von Mises distribution, the bivariate Cosine distribution \citep{mardia_protein_2007} and, more generally, a multivariate extension of the Cosine distribution \citep{mardia_multivariate_2008}. 
This is an unpleasant drawback of sine-skewed models, since the von Mises and Cosine models are among the most frequently used symmetric distributions on the unit circle $\mathcal{T}_1$ and unit torus $\mathcal{T}_2$.
Another disadvantage of sine-skewing lies in the condition $\sum_{i=1}^d \left| \lambda_i \right| \leq1$ as this restricts the parameter space and complicates numerical parameter estimation. These drawbacks represent an issue especially because of the strong need for good flexible models on the torus handling asymmetric data \citep{Mardia13}.

In the present paper, we therefore propose a new symmetry-modulating approach that overcomes these issues whilst keeping the appealing properties of sine-skewing. In particular, our new model, which we term \emph{arctan-skewed},  has a simple random number generation mechanism, achieves the same range of skewness as its sine-skewed antecedent without being subjected to the cumbersome parameter restriction condition, and is free from Fisher Information singularities for all known symmetric densities $f_0$ from the literature; see Section~\ref{sec: model} for details. We explain inferential properties in Section~\ref{sec: inference}, namely maximum likelihood estimation procedures and optimal tests for symmetry against our arctan-skewed alternatives.  
Through a real data analysis in Section~\ref{sec: real} we moreover show the superior fitting ability of our new model. 
We conclude the paper with final comments in Section~\ref{sec: final}.
The proofs of all results can be found in Section~\ref{sec: technical details}. 
The observed Fisher Information matrix and some simulations related to the tests of Section~\ref{sec: optimtest} are provided in the Supplementary Material.

\section{The arctan-skewed model and its main properties\label{sec: model}}

\subsection{Our proposal}\label{sec: ourprop}
As mentioned in the Introduction, our goal is to propose an alternative to sine-skewing which retains most of its attractive properties whilst overcoming its two main shortcomings. To this end, we start from the density~\eqref{eq: sine_skewed_density} and modify it accordingly. For the sake of simplicity, we will drop the parameter $\varthetab$ from the symmetric density and just write $f_0(\thetab-\mub)$. Our symmetric densities $f_0$ all belong to the class of functions
\begin{align}\label{eq: def: F}
    \mathcal{F} := \left\{\begin{array}{ll}
& f_0(\theta_1 + 2\pi k_1, \ldots, \theta_d + 2\pi k_d) = f_0(\theta_1, \ldots, \theta_d) \ \forall k_i\in\Z, i=1,\ldots,d,\\
f_0(\thetab): & f_0(\thetab) = f_0(-\thetab)\ \forall \thetab \in \mathcal{T}_d, \\
& \int_{\mathcal{T}_d} f_0(\thetab) {\mathrm d} \thetab = 1.
\end{array} \right\},
\end{align}
which ensures that we have symmetric and $2\pi$-periodic densities.

In order to lift the restriction $\sum_{j=1}^{d} |\lambda_j| \leq 1$, we introduce a bounded function acting on $\lambda_j \sin (\theta_j-\mu_j)$. We choose the arctan function which varies within $[-\pi/2,\pi/2]$, yielding 
$$
     f_0 (\thetab - \mub) \left( 1 + \frac{2}{\pi d} \sum_{i=1}^d \arctan (\lambda_i \sin (\theta_i - \mu_i)) \right).
$$
The univariate version of this proposal, named Cauchy-skewed circular distribution, was first mentioned in \cite{Abe2022}.
This construction however still suffers from the second drawback, the singular Fisher information matrix at $\lambdab=\zerob$ when $f_0$ is the Cosine density. 
Indeed, the score function for skewness in the above proposal is the same as for the sine-skewed model (see \cite{schutte_singularity_2026} for more details). 
Therefore we introduce a modification inside the arctan function which keeps it an odd function in $\theta_i-\mu_i$. 
This leads us to our new \textbf{arctan-skewed densities} of the form 
\begin{equation}\label{eq: sum_model}
    f(\thetab - \mub; \lambdab) := f_0 (\thetab - \mub) \left( 1 + \frac{2}{\pi d} \sum_{i=1}^d \arctan \left( \frac{\lambda_i \sin (\theta_i - \mu_i)}{1 + \sin^2 (\theta_i - \mu_i)} \right) \right).
\end{equation}
Of course other choices than dividing by $1 + \sin^2 (\theta_i - \mu_i)$ are possible, in the same way as the famous skew-t distribution of \cite{azzalini2003distributions} could have been built differently (see the discussion at the end of their Section 4.1); as we will see, our choice leads to desirable properties. 
It is  motivated by the fact that  the derivative of $\kappa\mapsto{\rm arctan}(\kappa\sin(\theta))$ is $\sin(\theta)/(1+\sin^2(\theta))$ at $\kappa=1$, hence it corresponds to the rate of change of the distance between sine-skewing and its symmetric counterpart.
Other choices of $\kappa$ would be possible; we choose the smallest integer one that does not oversimplify our model.
Multiplying the sum with $2\pi/d$  ensures that the resulting function is positive when the components of the skewness parameter $\lambdab\in\R^d$ go to plus and/or minus infinity (depending on the sign of $\sin (\theta_i - \mu_i)$ for each $i$). Thus, since $$\frac{2}{\pi d} \sum_{i=1}^d \arctan \left( \frac{\lambda_i \sin (\theta_i - \mu_i)}{1 + \sin^2 (\theta_i - \mu_i)} \right) \geq -1$$ and  $\int_{\mathcal{T}_d}f_0 (\thetab)\arctan \left( \lambda \sin \theta_i / \left( 1 + \sin^2 \theta_i \right) \right) \textrm{d} \thetab =0$ for all $i\in\{1,\ldots,d\}$ by symmetry arguments, we readily have that \eqref{eq: sum_model} is a valid density on $\mathcal{T}_d$. The normalizing constant of the new, skewed distribution $f(\cdot)$ is the same as $f_0(\cdot)$ and requires no complex calculations. This can also be seen through the expression of the cdf. For $\thetab \in \mathcal{T}_d$ we have
\begin{equation}
    F ( \thetab;\lambdab ) = F_0(\thetab ) + \frac{2}{\pi d} \sum_{i=1}^d \int_{-\boldsymbol{\pi}}^{\thetab} f_0 (\xb ) \arctan \left( \frac{\lambda_i \sin (x_i )}{1 + \sin^2 (x_i )} \right) \mathrm{d} \xb,
\end{equation}
where $F_0(\cdot)$ is the cdf corresponding to $f_0$. Using the fact that the functions $\arctan(\cdot)$ and $\sin(\cdot)$ are asymmetric around 0, it holds that
\begin{equation*}
    F ( \boldsymbol{\pi};\lambdab ) - F ( -\boldsymbol{\pi};\lambdab ) = F_0 ( \boldsymbol{\pi} ) - F_0 ( -\boldsymbol{\pi} ) = 1.
\end{equation*}

When $d=1$, hence on the circle, our proposal inscribes itself within the general skewing setting proposed by \cite{umbach_building_2009}, who considered skew densities of the form
\begin{equation*}
    f_\mu (\theta) = 2 f(\theta - \mu) G( w(\theta - \mu)),
\end{equation*}
where $f$ and $g$ are circular densities which are symmetric about $0$ and $G( \theta ) =  \int_{-\pi}^\theta g(\alpha) \mathrm{d} \alpha$. 
Also, $w$ is an odd  function with $|w(\theta)| \leq \pi$ and periodic.
Our arctan-skewed densities correspond to   $G(\cdot)$ as the cdf of the uniform distribution, meaning $G(\theta) = \left(\pi + \theta\right)/2 \pi$, and to $
    w(\theta) = 2\arctan \left( \lambda \sin \theta / \left( 1 + \sin^2 \theta \right) \right)$
for $\lambda \in \R$. 

In order to illustrate the effect of our arctan-skewing on symmetric distributions, we provide various contour plots in Figures~\ref{fig:contour_plots_1} and~\ref{fig:contour_plots_2} for the bivariate case. More precisely, Figure~\ref{fig:contour_plots_1} shows the original symmetric densities and their arctan-skewed counterparts for 
\begin{itemize}
\item[$\bullet$] the bivariate wrapped Cauchy (bwC) distribution \citep{kato_mobius_2015} with density
\begin{align}
    f_{\text{bwC}} (\thetab; \mub, \xib, \rho) 
    = c(\xib, \rho) & \Bigl \{  c_0(\xib, \rho) - c_1(\xib, \rho) \cos(\theta_1 - \mu_1) - c_2(\xib, \rho) \cos(\theta_2 - \mu_2)\nonumber\\
    &  - c_3(\xib, \rho) \cos(\theta_1 - \mu_1) \cos(\theta_2 - \mu_2) - c_4(\xib, \rho) \sin(\theta_1 - \mu_1) \sin(\theta_2 - \mu_2) \Bigr\}^{-1} \label{eq: bivariate_wrapped_cauchy}
\end{align}
where $\mub \in \mathcal{T}_2$ and 
$c(\xib, \rho)$, $c_0(\xib, \rho)$, $c_1(\xib, \rho)$, $c_2(\xib, \rho)$, $c_3(\xib, \rho)$, $c_4(\xib, \rho)$ are functions of 
$\xib \in [0,1)^2$  and  $\rho \in (-1,1)$ (see \cite{kato_mobius_2015} for details);
\item[$\bullet$] the Sine distribution \citep{singh_probabilistic_2002} with density
\begin{equation*}
    f_{\text{S}} (\thetab; \mub, \kappab, \beta) = C_{\text{S}} \exp \left\{ \kappa_1 \cos(\theta_1 - \mu_1) + \kappa_2 \cos(\theta_2 - \mu_2) + \beta \sin(\theta_1 - \mu_1)\sin(\theta_2 - \mu_2) \right\}
\end{equation*}
with appropriate yet complicated normalizing constant $C_{\text{S}}$, $\mub \in \mathcal{T}_2$, $\kappab \in \R_{\geq 0}^2$ and $\beta\in\R$;
\item[$\bullet$] the Cosine distribution \citep{mardia_protein_2007} with density
\begin{equation*}
    f_{\text{C}} (\thetab; \mub, \kappab, \beta) = C_{\text{C}} \exp \left\{ \kappa_1 \cos(\theta_1 - \mu_1) + \kappa_2 \cos(\theta_2 - \mu_2) - \beta \cos(\theta_1 - \mu_1 - \theta_2 + \mu_2) \right\}
\end{equation*}
with appropriate yet complicated normalizing constant $C_{\text{C}}$, $\mub \in \mathcal{T}_2$, $\kappab \in \R_{\geq 0}^2$ and $\beta\in\R$.
\end{itemize}
In Figure~\ref{fig:contour_plots_2} we delve more into the arctan-skewing effect on the bivariate wrapped Cauchy density~\eqref{eq: bivariate_wrapped_cauchy} by varying the skewness parameter. Both figures combined provide a visual aid to better grasp how arctan-skewing gradually turns  symmetric toroidal distributions into asymmetric ones.

\begin{figure}[htbp]
    \centering

    \begin{subfigure}{0.48\textwidth}
        \centering
        \includegraphics[height=5cm, width=\textwidth,
            trim=0cm 1cm 0cm 2cm, clip]{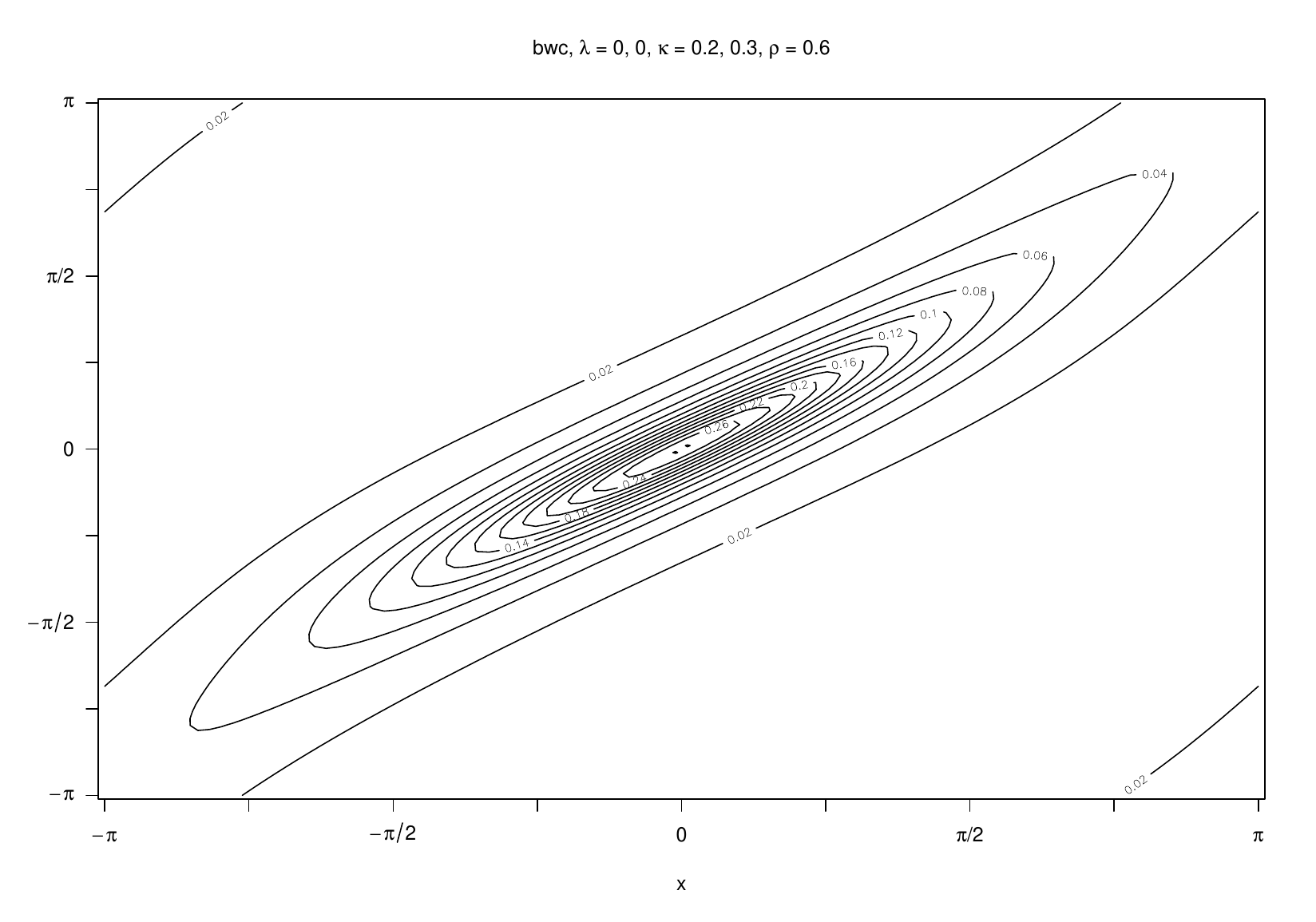}
        \caption{Bivariate wrapped Cauchy density with $\lambda_1, \lambda_2=0$.}
    \end{subfigure}
    \begin{subfigure}{0.48\textwidth}
        \centering
        \includegraphics[height=5cm, width=\textwidth,
            trim=0cm 1cm 0cm 2cm, clip]{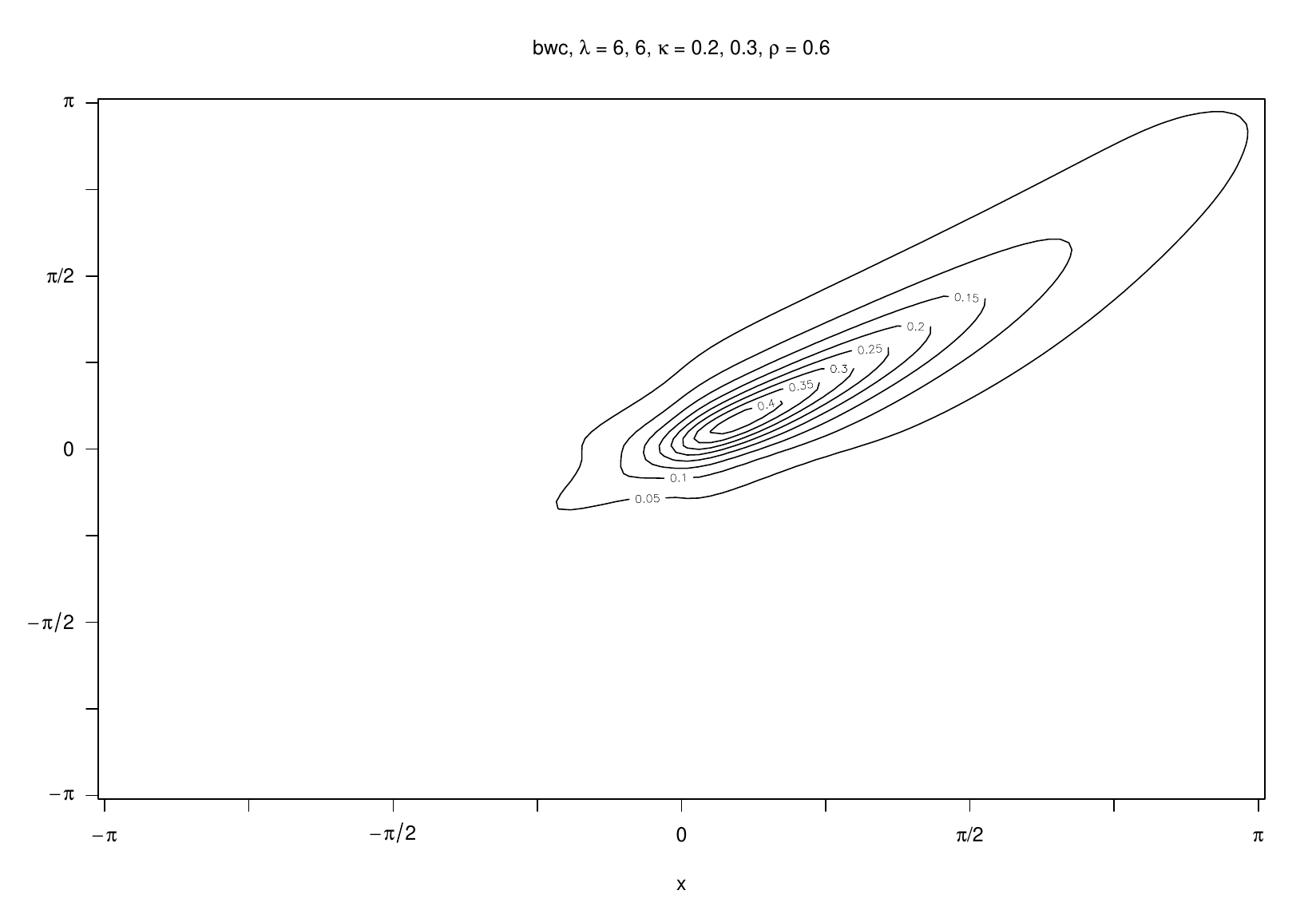}
        \caption{Bivariate wrapped Cauchy density with $\lambda_1, \lambda_2=6$.}
    \end{subfigure}

    \begin{subfigure}{0.48\textwidth}
        \centering
        \includegraphics[height=5cm, width=\textwidth,
            trim=0cm 1cm 0cm 2cm, clip]{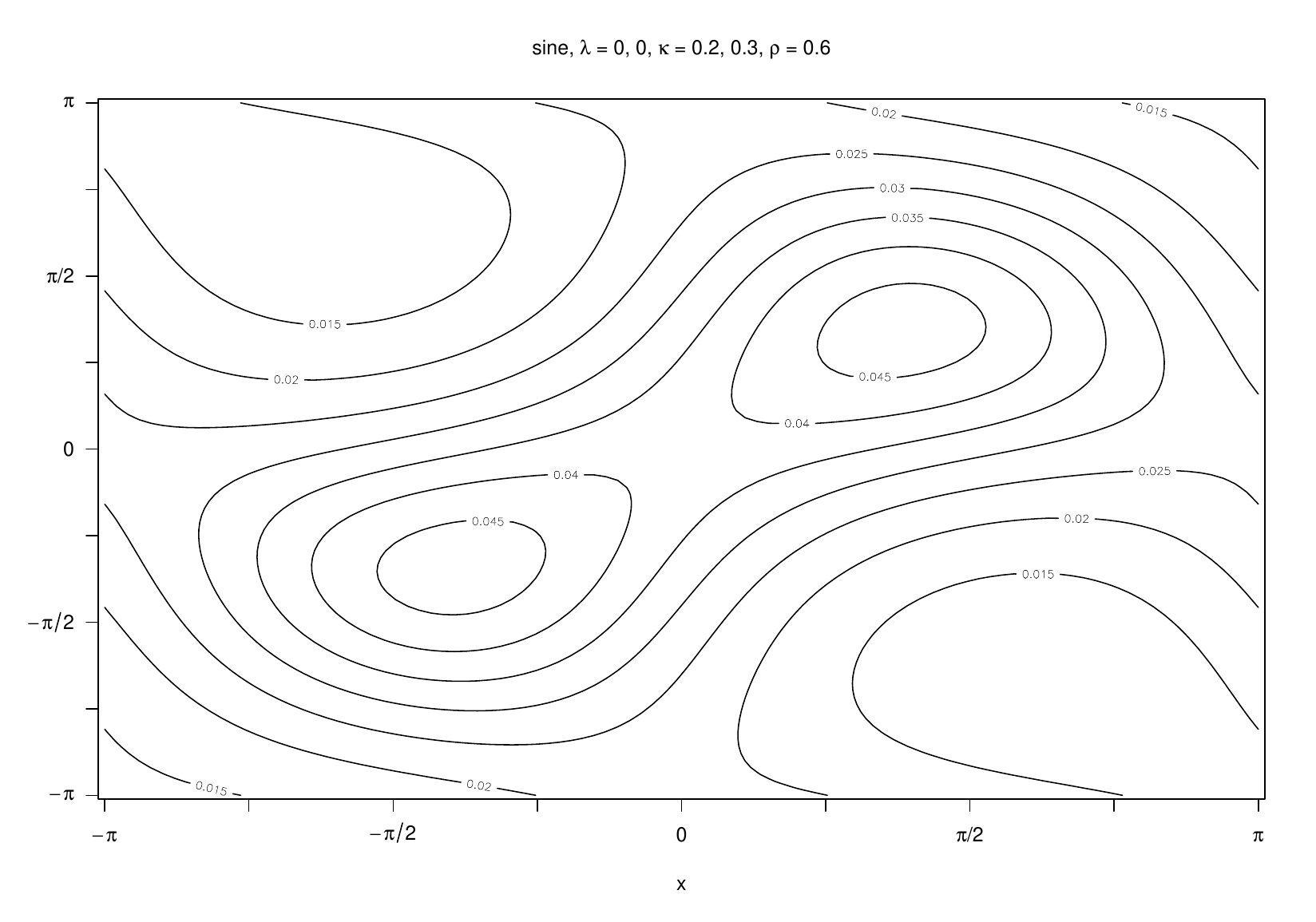}
        \caption{Sine density with $\lambda_1, \lambda_2=0$.}
    \end{subfigure}
    \begin{subfigure}{0.48\textwidth}
        \centering
        \includegraphics[height=5cm, width=\textwidth,
            trim=0cm 1cm 0cm 2cm, clip]{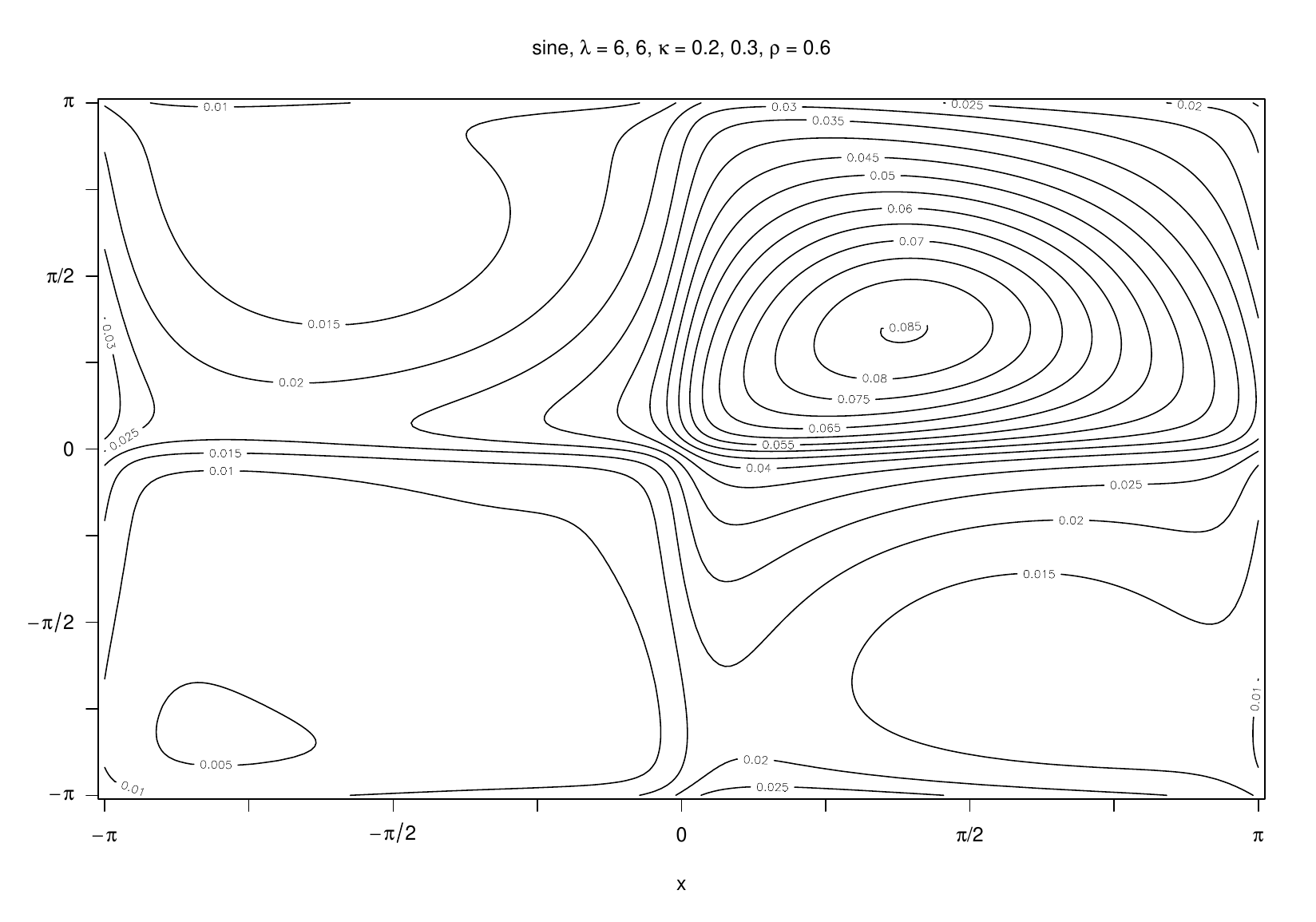}
        \caption{Sine density with $\lambda_1, \lambda_2=6$.}
    \end{subfigure}

    \begin{subfigure}{0.48\textwidth}
        \centering
        \includegraphics[height=5cm, width=\textwidth,
            trim=0cm 1cm 0cm 2cm, clip]{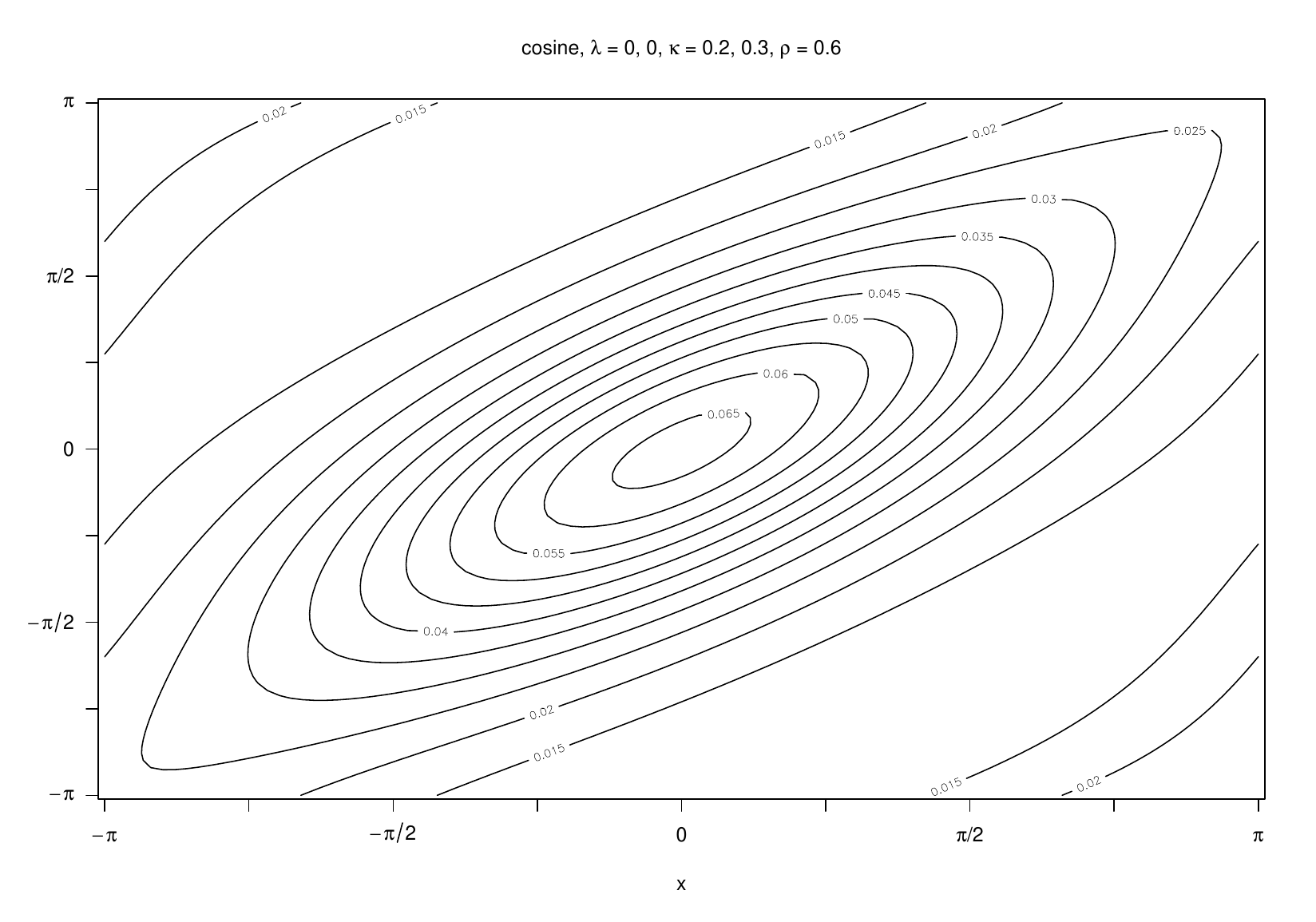}
        \caption{Cosine density with $\lambda_1, \lambda_2=0$.}
    \end{subfigure}
    \begin{subfigure}{0.48\textwidth}
        \centering
        \includegraphics[height=5cm, width=\textwidth,
            trim=0cm 1cm 0cm 2cm, clip]{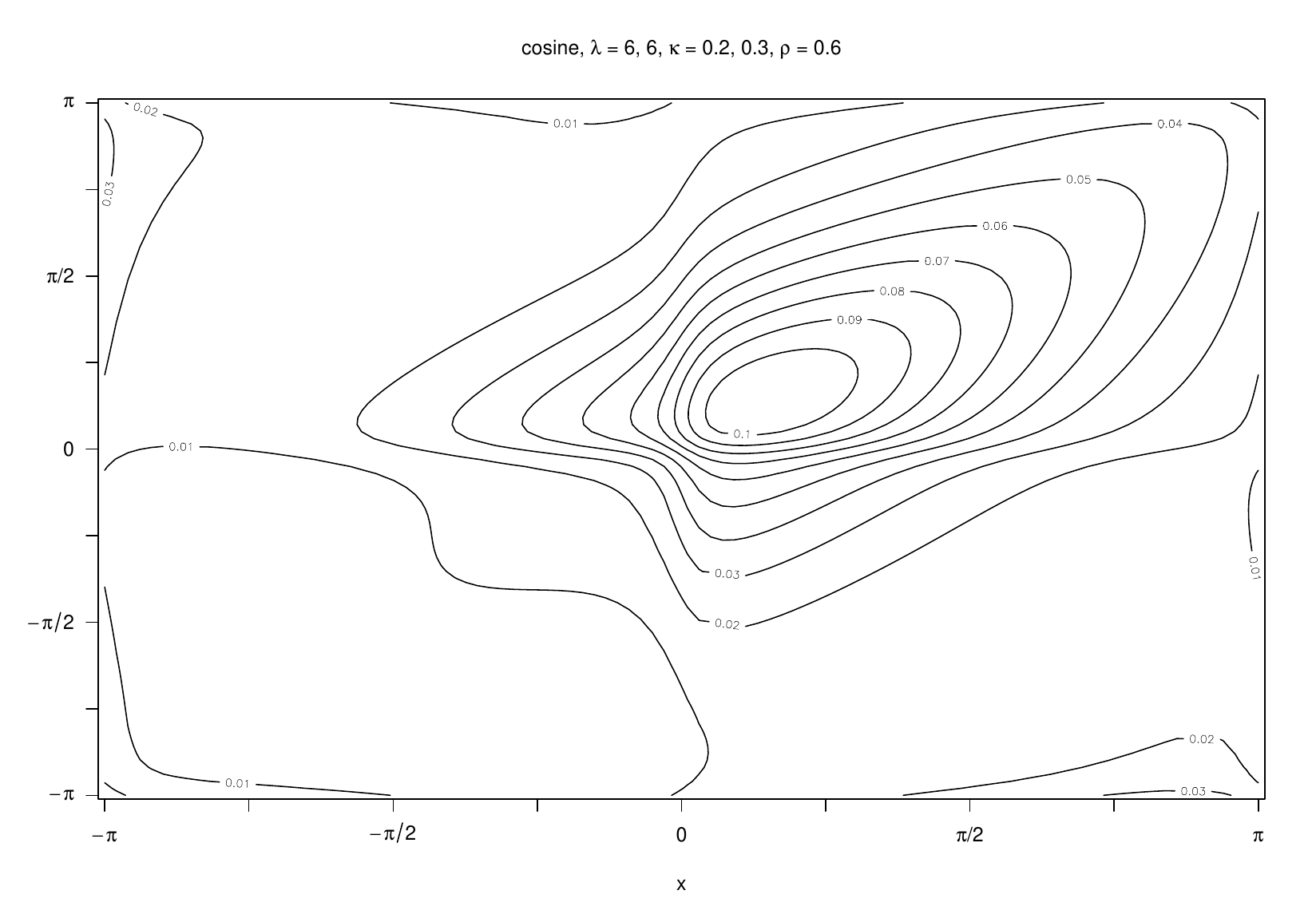}
        \caption{Cosine density with $\lambda_1, \lambda_2=6$.}
    \end{subfigure}

    \caption{Contour plots of arctan-skewed bivariate wrapped Cauchy, Sine and Cosine density with parameters $\xi_1 = \kappa_1 = 0.2$, $\xi_2 = \kappa_2 = 0.3$, $\rho = 0.6$ and skewness parameters $\lambda_1, \lambda_2$ as indicated under each plot.}
    \label{fig:contour_plots_1}
\end{figure}

\begin{figure}[htbp]
    \centering

    \begin{subfigure}{0.48\textwidth}
        \centering
        \includegraphics[height=5cm, width=\textwidth, trim=0cm 1cm 0cm 2cm, clip]{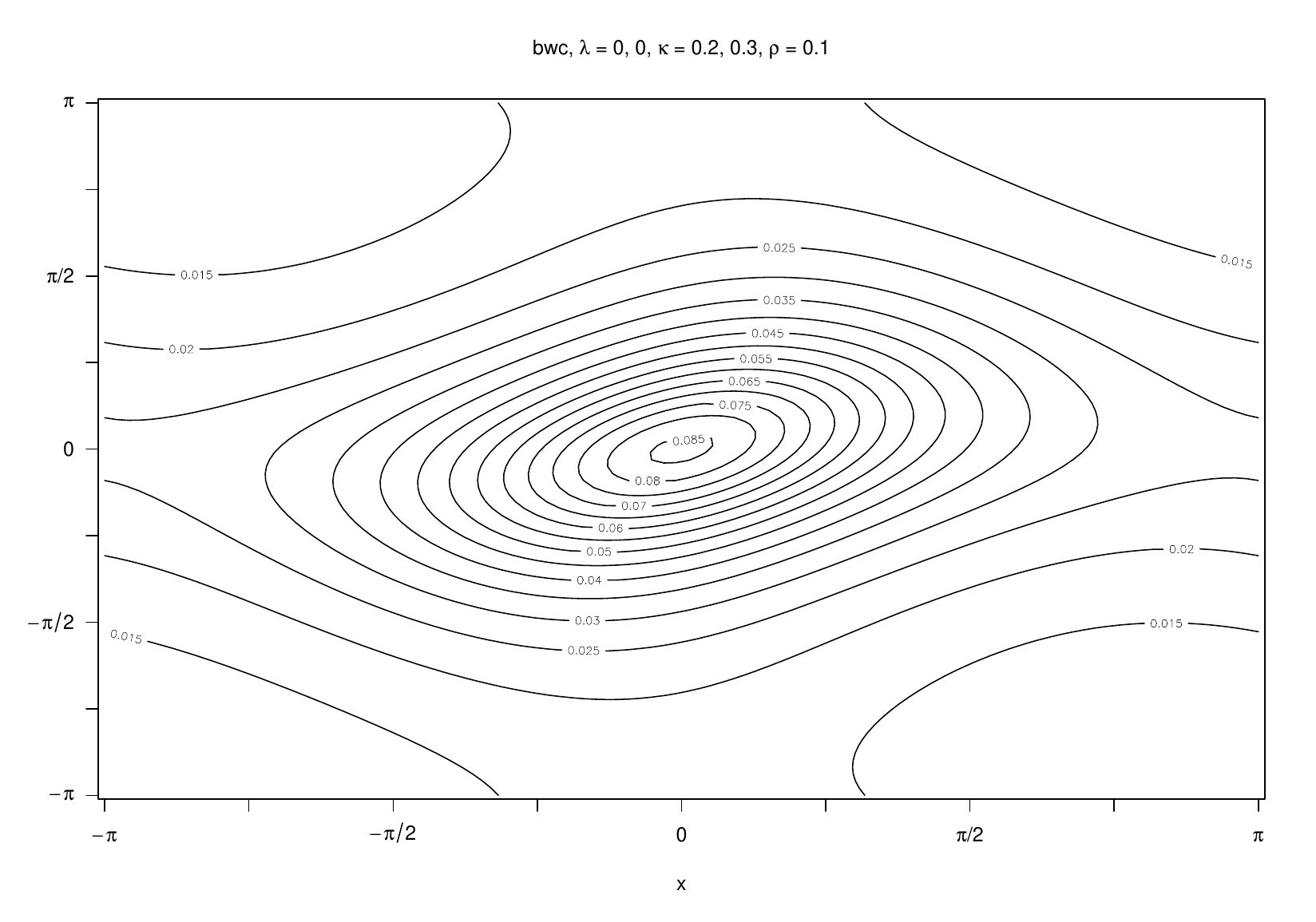}
        \caption{$\lambda_1 = \lambda_2 = 0$}
    \end{subfigure}
    \begin{subfigure}{0.48\textwidth}
        \centering
        \includegraphics[height=5cm, width=\textwidth, trim=0cm 1cm 0cm 2cm, clip]{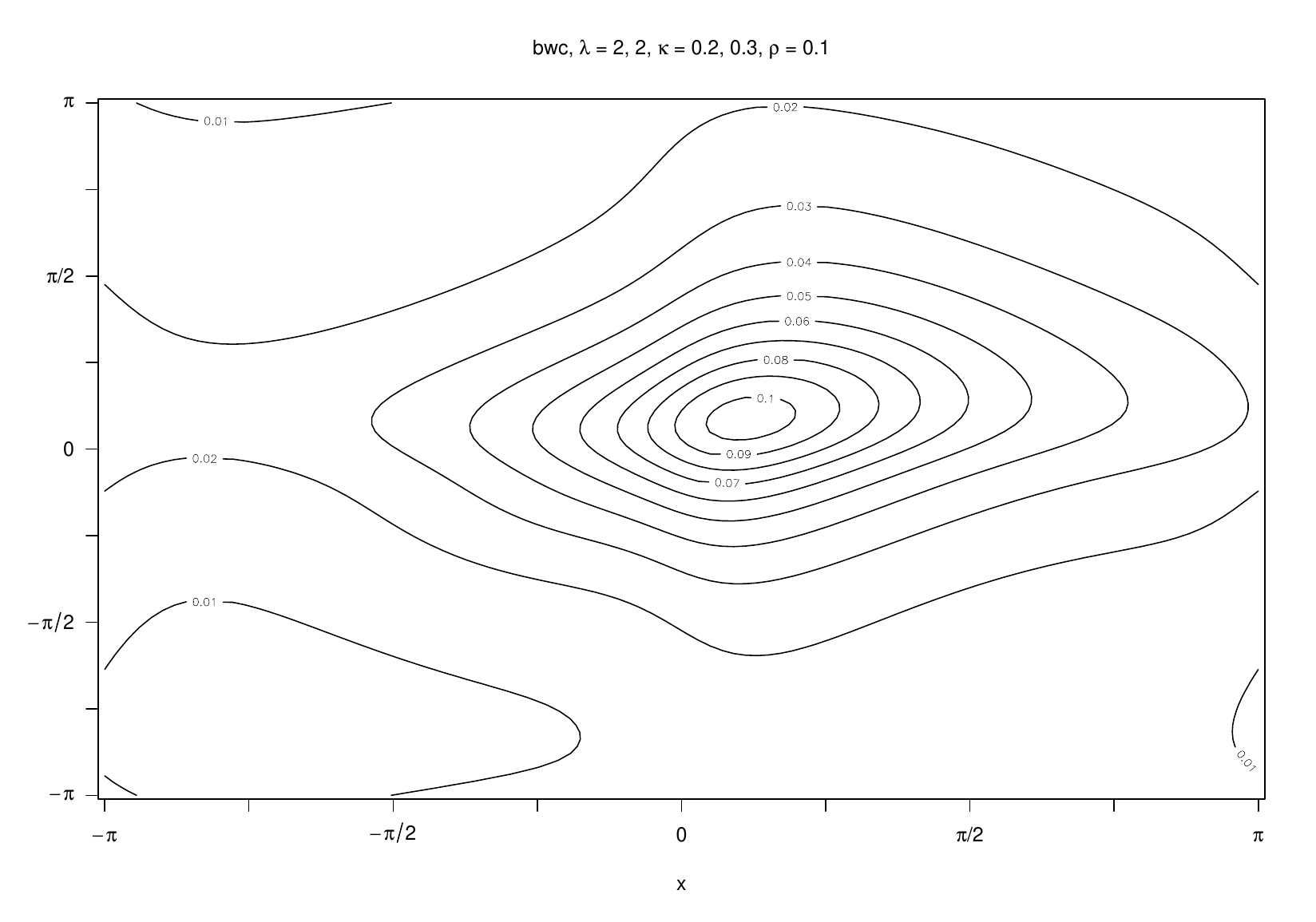}
        \caption{$\lambda_1 = \lambda_2 = 2$}
    \end{subfigure}

    \begin{subfigure}{0.48\textwidth}
        \centering
        \includegraphics[height=5cm, width=\textwidth, trim=0cm 1cm 0cm 2cm, clip]{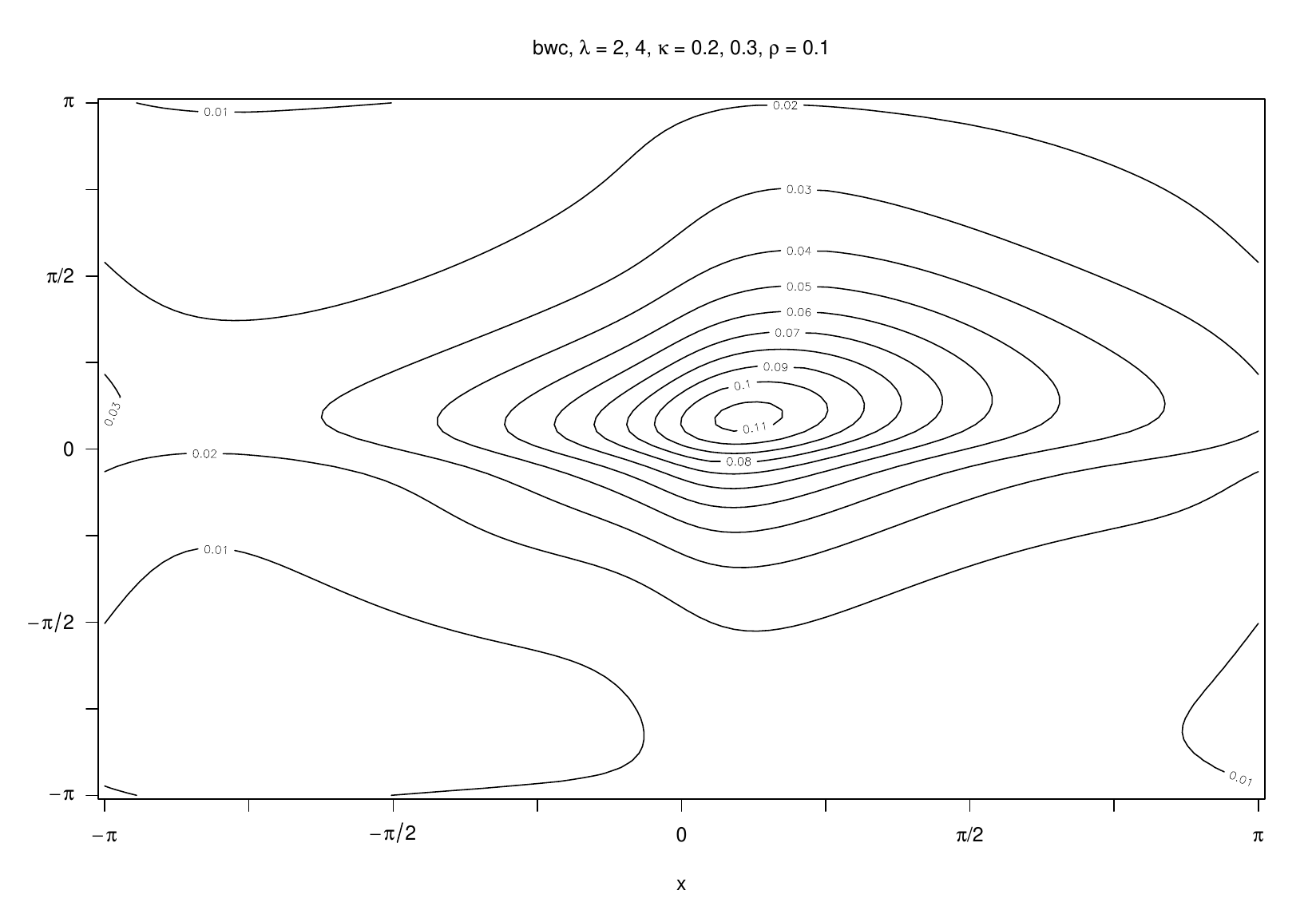}
        \caption{$\lambda_1 = 2$, $\lambda_2 = 4$}
    \end{subfigure}
    \begin{subfigure}{0.48\textwidth}
        \centering
        \includegraphics[height=5cm, width=\textwidth, trim=0cm 1cm 0cm 2cm, clip]{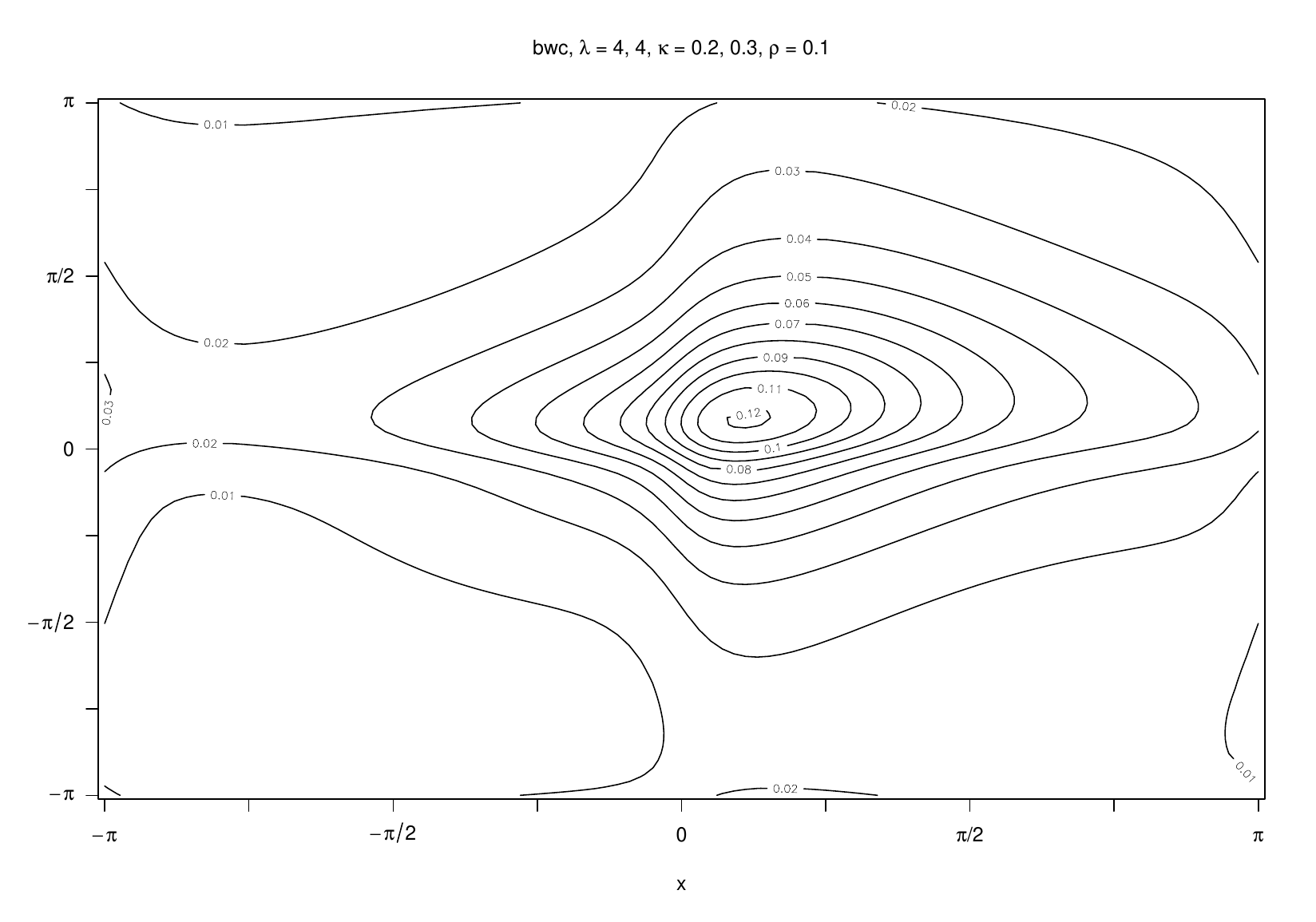}
        \caption{$\lambda_1 = \lambda_2 = 4$}
    \end{subfigure}
    \begin{subfigure}{0.48\textwidth}
        \centering
        \includegraphics[height=5cm, width=\textwidth, trim=0cm 1cm 0cm 2cm, clip]{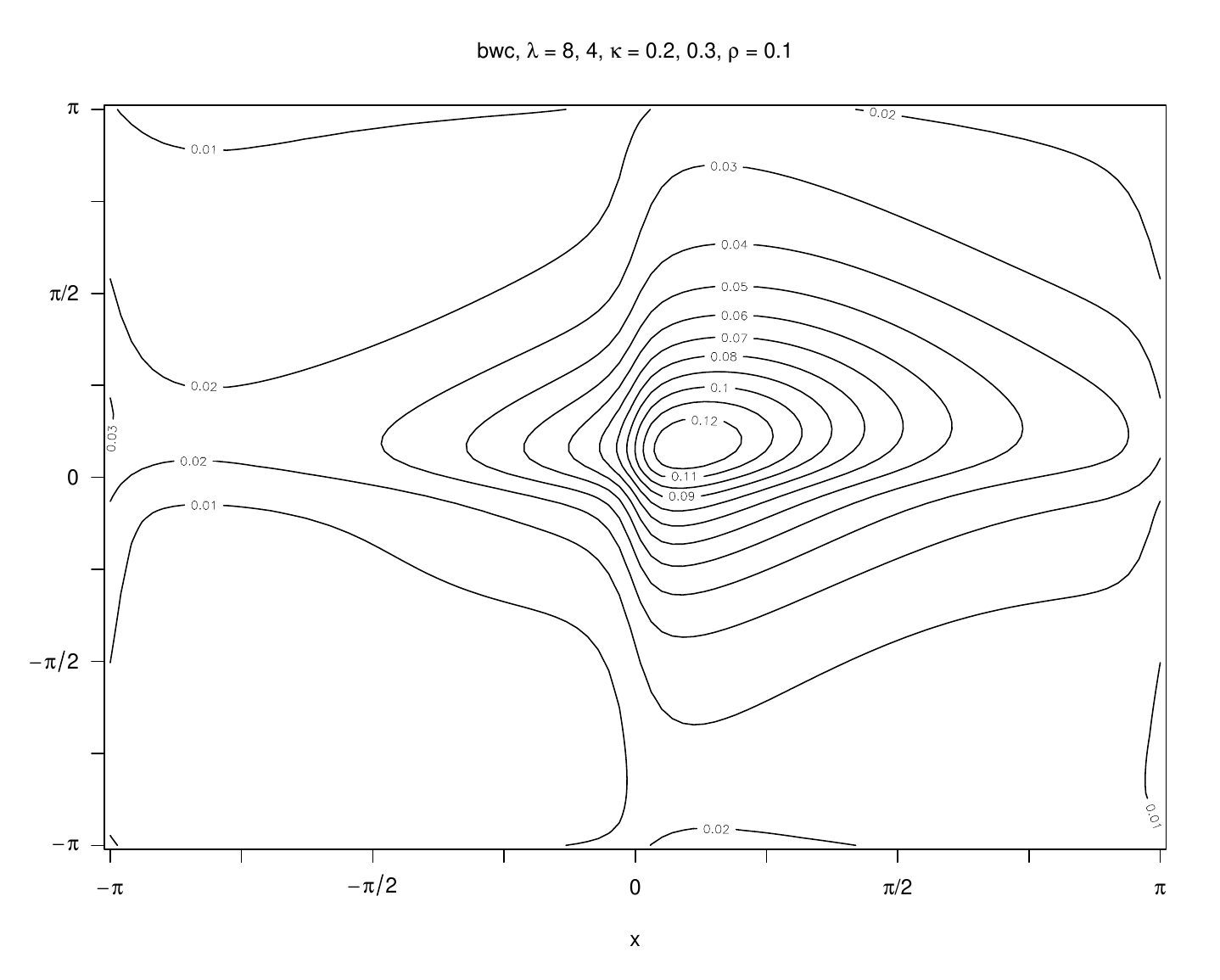}
        \caption{$\lambda_1 = 8$, $\lambda_2 = 4$}
    \end{subfigure}
    \begin{subfigure}{0.48\textwidth}
        \centering
        \includegraphics[height=5cm, width=\textwidth, trim=0cm 1cm 0cm 2cm, clip]{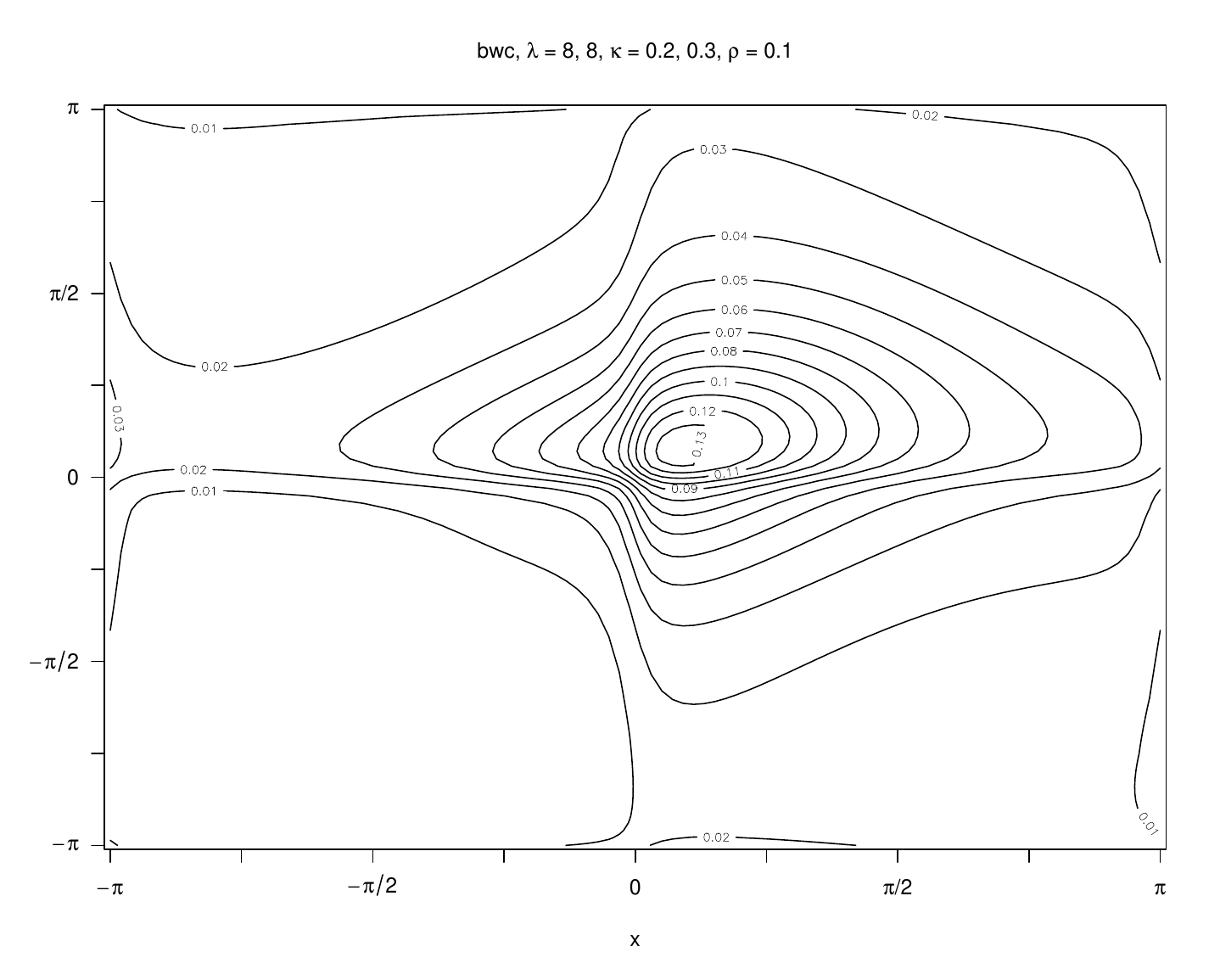}
        \caption{$\lambda_1 = \lambda_2 = 8$}
    \end{subfigure}

    \caption{Contour plots of arctan-skewed bivariate wrapped Cauchy densities with parameters $\xi_1 = 0.2$, $\xi_2 = 0.3$, $\rho = 0.1$ and skewness parameters $\lambda_1, \lambda_2$ as indicated under each plot.}
    \label{fig:contour_plots_2}
\end{figure}

In the rest of this section, we will investigate the main properties of our new model such as marginal and conditional distributions (Section~\ref{sec: margcond}), moment expressions (Section~\ref{sec: mom}) and ensuing measures for  dependence (Section~\ref{sec: depmeas}), as well as random number generation (Section~\ref{sec: randomgen}). However, before delving into these properties, we will first show in Section~\ref{sec: singularity} that our new model is indeed singularity-free for all known symmetric densities $f_0$ in the literature and explain why we did not opt for another possible arctan-skewing scheme in Section~\ref{sec: altprop}.

\subsection{Fisher Information singularity investigation\label{sec: singularity}}

In Theorem~\ref{thm: singularity} below, we give a complete characterization of which densities $f_0$ lead to a singular Fisher Information matrix. Before providing the statement of the theorem, we need to define the function
    \begin{equation} \label{eq: solution}
    h_0( \thetab {- \mub}) := f_0(\thetab {- \mub}) \prod^d_{i=1} \left( \frac{\sec ^2\left(\frac{\theta_i - \mu_i}{2}\right)-2 \sqrt{2}+2}{\sec ^2\left(\frac{\theta_i - \mu_i}{2}\right)+2 \sqrt{2}+2} \right)^{\delta_i}
    \end{equation}
for $\deltab = (\delta_1, \dots,\delta_d)^\intercal \in \R^d \setminus \{ \zerob \}$. 

\begin{theorem}\label{thm: singularity}
 Suppose that the mapping $\thetab \mapsto f_0(\thetab)$ is differentiable on $\mathcal{T}_d$.
The Fisher Information matrix of the arctan-skewed version of $f_0(\thetab {- \mub})$ in the vicinity of symmetry is singular if and only if there exists a vector $\alphab = \left( \alpha_1, \ldots, \alpha_d \right)^\intercal\in\R^d$ with at least one $\alpha_i \neq 0$, $i\in\{1,\ldots,d\}$, 
such that $h_0(\thetab {- \mub})$ defined in \eqref{eq: solution} 
satisfies
\begin{equation*}
    h_0(\thetab {- \mub} + t\alphab) = h_0(\thetab {- \mub})
\end{equation*}
for all $t\in \mathbb{R}$ and $\thetab \in [-\pi,\pi)^d $.
\end{theorem}

The proof is given in Section~\ref{sec: technical details}. Using Theorem~\ref{thm: singularity}, it holds that the arctan-skewed version of a symmetric base density $f_0(\cdot)$ has singularity issues if and only if it is of the form
\begin{equation*}
    f_0(\thetab-\mub) = h_0(\thetab-\mub) \prod^d_{i=1} \left( \frac{\sec ^2\left(\frac{\theta_i-\mu_i}{2}\right) + 2 \sqrt{2}+2}{\sec ^2\left(\frac{\theta_i-\mu_i}{2}\right) - 2 \sqrt{2}+2} \right)^{\delta_i}.
\end{equation*}
As no distributions of this form have been proposed in the literature, no known distributions combined with our arctan-skewing suffer from singularity of the Fisher Information matrix in the vicinity of symmetry, and hence our proposal overcomes the Fisher Information matrix singularity issue related to sine-skewing.
Note that setting $\deltab = \zerob$ corresponds to the symmetric density $f_0(\thetab-\mub)$ being singular, and thus we do not consider this case in the definition of~\eqref{eq: solution}.

\subsection{An alternative model: product arctan-skewing}\label{sec: altprop}

A natural alternative construction would be a product type model of the form
\begin{equation}\label{eq: product}
     f_0 (\thetab - \mub) \prod_{i=1}^d \left( 1 + \frac{2}{\pi} \arctan \left( \frac{\lambda_i \sin (\theta_i - \mu_i)}{1 + \sin^2 (\theta_i - \mu_i)} \right) \right).
\end{equation}
Such a product form could also exist for the sine-skewing approach but has never been discussed in the literature to the best of the authors' knowledge. However it has often been orally debated at conferences. Hence our discussion here  about our product model~\eqref{eq: product} will also shed light on a product sine-skewing option. 

It is clear that~\eqref{eq: product} is positive, however it does not integrate to 1 like our original proposal~\eqref{eq: sum_model} and a normalizing constant calculation is required. Moreover, random variate generation cannot be done in the same simple way as for our arctan-skewing (see Section~\ref{sec: randomgen}), likely requiring the use of acceptance-rejection algorithms. These are two clear disadvantages of the product form~\eqref{eq: product} compared to~\eqref{eq: sum_model}.  Regarding the Fisher Information singularity problem, immediate calculations yield that the score function of \eqref{eq: product} in the vicinity of symmetry corresponds to
\begin{equation} \label{eq: score_prod}
\left( 
         -\frac{\frac{\partial}{\partial \theta_1}f_0(\thetab-\mub)}{f_0(\thetab-\mub)},
        \dots,
        -\frac{\frac{\partial}{\partial \theta_d}f_0(\thetab-\mub)}{f_0(\thetab-\mub)},
        \frac{2}{\pi} \frac{\sin(\theta_{1}-\mu_1)}{ 1+ \sin^2(\theta_{1}-\mu_1)},
        \dots,
        \frac{2}{\pi} \frac{\sin(\theta_{d}-\mu_d)}{ 1 + \sin^2(\theta_{d}-\mu_d)} 
    \right)^\intercal.
\end{equation}
The score vector \eqref{eq: score_prod} is the same as \eqref{eq: score}, up to multiplying by $1/d$ the last $d$ terms.
This means that the result proved in Theorem~\ref{thm: singularity} also holds for the product model.

Summarizing, the product idea retains positive properties of our proposal~\eqref{eq: sum_model}, meaning the absence of any
Fisher Information singularity issues and  unrestricted skewness parameter space,  but suffers from the drawbacks of normalizing constant needing to be determined and a more involved random number generation procedure. For these reasons we opted for our proposal~\eqref{eq: sum_model}.

\subsection{Marginals and conditionals}\label{sec: margcond}

Let $\Theta = (\Theta_1, \ldots , \Theta_d)^\intercal$ follow a $d$-dimensional arctan-skewed toroidal distribution and let $m$ be a subset of $\{1, \ldots , d\}$ of size $M$ and $m^c$ its complementary. Assume without loss of generality that $\mub = 0$ and denote by $f_m$ and $f_{0;m}$  the marginal skewed and symmetric densities associated with the subvector $\Thetab_m$  of $\Thetab$ corresponding to the subset $m$ of indices. The marginal distribution of the subvector $\Thetab_m$ then is given in the following theorem, whose proof is given in Section~\ref{sec: technical details}.

\begin{theorem}\label{theo:margins}
Let $f_0\in\mathcal{F}$. If $f_0(\thetab_m, \thetab_{m^c}) = f_0(\thetab_m, -\thetab_{m^c})$, then
\begin{equation*}
    f_m(\thetab_m)=f_{0;m} (\thetab_m) \left( 1 + \frac{2}{\pi d} \sum_{i\in m} \arctan \left( \frac{\lambda_i \sin (\theta_i)}{1 + \sin^2 (\theta_i)} \right) \right)
\end{equation*}
for $\thetab_m\in\mathcal{T}_M$ and $\thetab_{m^c}\in\mathcal{T}_{d-M}$.
\end{theorem}

The marginal distribution thus is of a simple form and almost arctan-skewed (the ratio $\frac{2\pi}{d}$ would need to be $\frac{2\pi}{M}$). Let us consider a slightly more general version of our proposed transformation, namely
\begin{equation}\label{eq: sum_model2}
    f(\thetab - \mub;\lambdab) = S_{\{1,\ldots,d\}, d, \lambdab}(1) f_0 (\thetab - \mub) 
\end{equation}
where
\begin{equation*}
    S_{A, d, \lambdab}(c) = 1 + \frac{2}{\pi d c} \sum_{i \in A} \arctan \left( \frac{\lambda_i \sin (\theta_i - \mu_i)}{1 + \sin^2 (\theta_i - \mu_i)} \right)
\end{equation*}
for a set $A \subseteq \{1,\ldots,d\}$ and $c \geq 1$ can depend on the observations, meaning that $c = c(\theta_1,\ldots, \theta_d)$.
Setting $c=1$ is the largest value that ensures that the density is non-negative. For this family, the result from Theorem~\ref{theo:margins} reads
$$
f_m(\thetab_m)=f_{0;m} (\thetab_m) \left( 1 + \frac{2}{\pi d} \sum_{i\in m} \arctan \left( \frac{\lambda_i \sin (\theta_i)}{1 + \sin^2 (\theta_i)} \right) \right)= S_{m, d, \lambdab}(1) f_{0;m} (\thetab_m)
$$
and hence we obtain closure under marginalization.

We now consider  conditional distributions for~\eqref{eq: sum_model}.

\begin{theorem}\label{theo:conds}
Let $f_0\in\mathcal{F}$. If $f_0(\thetab_m, \thetab_{m^c}) = f_0(\thetab_m, -\thetab_{m^c})$, then
\begin{equation*}
    f_{m|m^c}(\thetab_m| \thetab_{m^c}) =f_{0;m|m^c} (\thetab_{m} | \thetab_{m^c}) \left(1 + \frac{\frac{2}{\pi d} \sum_{i \in m} \arctan \left( \frac{\lambda_i \sin (\theta_i)}{1 + \sin^2 (\theta_i)} \right)}{1 + \frac{2}{\pi d} \sum_{i\in m^c} \arctan \left( \frac{\lambda_i \sin (\theta_i)}{1 + \sin^2 (\theta_i)} \right)} \right)
\end{equation*}
for $\thetab_m\in\mathcal{T}_M$ and $\thetab_{m^c}\in\mathcal{T}_{d-M}$.
\end{theorem}
The proof is given in Section~\ref{sec: technical details}. 
As for the marginals, we obtain a simple tractable expression, but the model is not closed under conditioning. Adopting the same more general model~\eqref{eq: sum_model2}, one can see that
\begin{align*}
    f_{m|m^c}(\thetab_m| \thetab_{m^c}) 
    & = f_{0;m|m^c} (\thetab_{m} | \thetab_{m^c}) \left(1 + \frac{\frac{2}{\pi d} \sum_{i \in m} \arctan \left( \frac{\lambda_i \sin (\theta_i)}{1 + \sin^2 (\theta_i)} \right)}{1 + \frac{2}{\pi d} \sum_{i\in m^c} \arctan \left( \frac{\lambda_i \sin (\theta_i)}{1 + \sin^2 (\theta_i)} \right)} \right) \\
    & = f_{0;m|m^c} (\thetab_{m} | \thetab_{m^c}) \left(1 + \frac{2}{\pi d S_{m^c, d, \lambdab}(1)} \sum_{i \in m} \arctan \left( \frac{\lambda_i \sin (\theta_i)}{1 + \sin^2 (\theta_i)} \right) \right) \\
    & = S_{m, d, \lambdab} \left( S_{m^c, d, \lambdab}(1) \right) f_{0;m|m^c} (\thetab_{m} | \thetab_{m^c}),
\end{align*}
which thus corresponds to closure under both marginalization and conditioning.

\subsection{Trigonometric moments}\label{sec: mom}

We denote the characteristic function  by $\{\phi_{\pb}: p_j = 0, \pm 1, \pm 2,\ldots ; j=1,\ldots,d \}$ where 
\begin{equation} \label{eq: Fourier_theory}
    \phi_{\pb} 
    := \E \left[ e^{i \pb^\intercal \Thetab} \right] 
    = \E \left[ \cos( \pb^\intercal \Thetab) \right] + i \E \left[ \sin( \pb^\intercal \Thetab) \right]
    =: \alpha_{\pb} + i\beta_{\pb},
\end{equation}
and we adopt the notation $\alpha_{\pb}^0$ and $\beta_{\pb}^0$ for the  trigonometric moments of the original symmetric density $f_0$. The trigonometric moments for our arctan-skewed distribution are then given in the following result.
\begin{theorem}\label{Theor: moments}
Let $f_0\in\mathcal{F}$. The trigonometric moments of the arctan-skewed model correspond to 
\begin{align*}
    & \alpha_{\pb}=\alpha^0_{\pb}, \\
    & \beta_{\pb} =\frac{1}{\pi d} \sum_{i=1}^d \sum_{n=1}^\infty  b_n (\lambda_i) \left( \alpha^0_{\pb - n \eb_i} - \alpha^0_{\pb + n \eb_i} \right),
\end{align*}
where   $\eb_i = (0, \ldots, 0, 1, 0, \ldots, 0)^\intercal$ denotes the vector with all elements equal to 0 except position $i$ which is equal to $1$ and
\begin{equation}\label{Fourcoeff}
    b_n (\lambda_i) = \frac{1}{\pi} \int_{-\pi}^\pi \arctan \left( \frac{\lambda_i \sin (\theta_i)}{1 + \sin^2 (\theta_i)} \right) \sin \left( n \theta_i \right) \mathrm{d} \theta_i
    \end{equation}
is the Fourier sine series coefficient of the function $ \arctan \left( \lambda_i \sin (\theta_i)/\left(1 + \sin^2 (\theta_i)\right) \right)$.
\end{theorem}

The proof of this theorem is given in Section~\ref{sec: technical details}. The moments can also be calculated numerically for a chosen symmetric density $f_0(\cdot)$.
In this case, it is more convenient to use  form \eqref{eq: cosine_moments} of the moment expression.

We conclude this section by providing the toroidal mean, concentration, skewness and kurtosis. 
The elements of the toroidal mean $\mub_m = \left( \mu_{m;1}, \ldots, \mu_{m;d} \right)^\intercal$ are given by
\begin{equation*}
    \mu_{m;s} = \atan 2 \left( \beta_{\eb_s}, \alpha_{\eb_s} \right)
    = \atan 2 \left( \frac{1}{\pi d} \sum_{i=1}^d \sum_{n=1}^\infty  b_n (\lambda_i) \left( \alpha^0_{\eb_s - n \eb_i} - \alpha^0_{\eb_s + n \eb_i} \right), \alpha^0_{\eb_s} \right)
\end{equation*}
for $s \in \{ 1,\ldots, d\}$, where $\atan 2(y, x)$ computes the principal value of the argument function applied to the complex number $x+iy$.
Note that when $\lambdab = \boldsymbol{0}$, the mean direction coincides with the mean of the symmetric density $f_0(\cdot)$.
The elements of the concentration vector $\Rb = \left( R_{1}, \ldots, R_{d} \right)^\intercal$ and the variance $\Vb = \left( V_1, \ldots, V_d \right)^\intercal$ are given by 
\begin{align*}
    R_{s} & = \sqrt{\alpha_{\eb_s}^2 + \beta_{\eb_s}^2}
    = \sqrt{\left( \alpha^0_{\eb_s} \right)^2 + \left( \frac{1}{\pi d} \sum_{i=1}^d \sum_{n=1}^\infty  b_n (\lambda_i) \left( \alpha^0_{\eb_s - n \eb_i} - \alpha^0_{\eb_s + n \eb_i} \right) \right)^2 }, \quad
    V_s  = 1 - R_{s}
\end{align*}
for $s \in \{1, \ldots, d\}$. The measures of circular skewness, $\boldsymbol{\mathcal{S}} = \left( \mathcal{S}_1, \ldots, \mathcal{S}_d \right)^\intercal$, and kurtosis, $\boldsymbol{\mathcal{K}} = \left( \mathcal{K}_1, \ldots, \mathcal{K}_d \right)^\intercal$, are given by \citep[Section 3.4.2]{mardia_directional_2000}
\begin{align*}
    \mathcal{S}_s = \frac{\Bar{\beta}_{2\eb_s}}{V_s^{3/2}}, \qquad 
    \mathcal{K}_s = \frac{\Bar{\alpha}_{2\eb_s} - R_{s}^4}{V_s^{2}},
\end{align*}
for $s\in\{1,\ldots,d\}$, where $\Bar{\alpha}_{\pb}$ and $\Bar{\beta}_{\pb}$ are the trigonometric moments about the mean direction $\mub_m$, given by $\Bar{\alpha}_{\pb} = \E[\cos( \pb^\intercal (\Thetab - \mub_m))]$ and $\Bar{\beta}_{\pb} = \E[\sin( \pb^\intercal (\Thetab - \mub_m))]$.
It holds that 
\begin{align*}
    \Bar{\alpha}_{2\eb_s} 
    & = \E \left[\cos \bigr( 2\eb_s^\intercal (\Thetab - \mub_m) \bigl) \right]
    = \E \left[ \cos \bigl( 2 (\Theta_s - \mu_{m;s}) \bigr) \right] 
    = \E \left[ \cos ( 2 \Theta_s) \cos (2\mu_{m;s}) + \sin ( 2 \Theta_s) \sin (2\mu_{m;s}) \right] \\
    & = \left( \cos^2 (\mu_{m;s}) - \sin^2 (\mu_{m;s}) \right) \E \left[ \cos ( 2 \Theta_s) \right] + 2 \sin (\mu_{m;s}) \cos (\mu_{m;s}) \E \left[ \sin ( 2 \Theta_s) \right] \\
    & = \frac{\alpha^2_{\eb_s} - \beta^2_{\eb_s}}{\alpha^2_{\eb_s} + \beta^2_{\eb_s}} \alpha_{2\eb_s} 
    + \frac{2 \alpha_{\eb_s} \beta_{\eb_s}}{\alpha^2_{\eb_s} + \beta^2_{\eb_s}} \beta_{2\eb_s} 
    = \frac{1}{R_{s}^2} \left[ \left( \alpha^2_{\eb_s} - \beta^2_{\eb_s} \right) \alpha_{2\eb_s} 
    + 2 \alpha_{\eb_s} \beta_{\eb_s} \beta_{2\eb_s} \right]
\end{align*}
and similarly,
\begin{align*}
    \Bar{\beta}_{2\eb_s} 
    & = \E \left[\sin \bigr( 2\eb_s^\intercal (\Thetab - \mub_m) \bigl) \right]
    = \E \left[ \sin \bigl( 2 (\Theta_s - \mu_{m;s}) \bigr) \right] 
    = \E \left[ \sin ( 2 \Theta_s) \cos (2\mu_{m;s}) - \cos ( 2 \Theta_s) \sin (2\mu_{m;s}) \right] \\
    & = \left( \cos^2 (\mu_{m;s}) - \sin^2 (\mu_{m;s}) \right) \E \left[ \sin ( 2 \Theta_s) \right] - 2 \sin (\mu_{m;s}) \cos (\mu_{m;s}) \E \left[ \cos ( 2 \Theta_s) \right] \\
    & = \frac{\alpha^2_{\eb_s} - \beta^2_{\eb_s}}{\alpha^2_{\eb_s} + \beta^2_{\eb_s}} \beta_{2\eb_s} 
    - \frac{2 \alpha_{\eb_s} \beta_{\eb_s}}{\alpha^2_{\eb_s} + \beta^2_{\eb_s}} \alpha_{2\eb_s} 
    = \frac{1}{R_{s}^2} \left[ \left( \alpha^2_{\eb_s} - \beta^2_{\eb_s} \right) \beta_{2\eb_s} 
    - 2 \alpha_{\eb_s} \beta_{\eb_s} \alpha_{2\eb_s} \right].
\end{align*}

Thanks to the expression for the measure of skewness, we can compare the range of skewness of our arctan-skewed model to that of the original sine-skewed model. Both models share in common that the cosine moment corresponds to the cosine moment of the original symmetric density $f_0$, hence the skewness measure $\mathcal{S}_s$ is only affected by the sine moments of each model. As $-\pi/2 \leq \arctan \left( \cdot \right) \leq \pi/2$, straightforward calculations for the arctan-skewed model lead to
\begin{align*}
    \beta_{\pb} 
    & = \frac{2}{\pi d} \sum_{i=1}^d \int_{\mathcal{T}_d} \sin(\pb^\intercal \thetab) f_0 (\thetab)  \arctan \left( \frac{\lambda_i \sin (\theta_i)}{1 + \sin^2 (\theta_i)} \right) \mathrm{d} \thetab \\
    & \leq \frac{2}{\pi d} \sum_{i=1}^d \int_{\mathcal{T}_d} \left| \sin(\pb^\intercal \thetab) \right| f_0 (\thetab) \frac{\pi}{2} \mathrm{d} \thetab 
    = \int_{\mathcal{T}_d} \left| \sin(\pb^\intercal \thetab) \right| f_0 (\thetab) \mathrm{d} \thetab =: \gamma^0_{\pb}
\end{align*}
and, similarly, $\beta_{\pb} \geq -\gamma^0_{\pb}$, so we can conclude that
\begin{equation}\label{eq: bound_on_cosine_moment}
    -\gamma^0_{\pb} \leq \beta_{\pb} \leq \gamma^0_{\pb}.
\end{equation}
Along the same lines, for the sine-skewed model it holds that
\begin{align*}
    \beta_{\pb}^{\rm SS} 
    & = \int_{\mathcal{T}_d} \sin(\pb^\intercal \thetab) f_0 (\thetab)  \left( 1+\sum_{j=1}^{d} \lambda_j \sin \theta_j \right) \mathrm{d} \thetab 
    = \sum_{j=1}^{d} \lambda_j \int_{\mathcal{T}_d} \sin(\pb^\intercal \thetab) f_0 (\thetab)  \sin \theta_j \mathrm{d} \thetab\\
    &\leq
    \sum_{j=1}^{d} |\lambda_j| \int_{\mathcal{T}_d} | \sin(\pb^\intercal \thetab) | f_0 (\thetab) \mathrm{d} \thetab \leq \gamma^0_{\pb}
\end{align*}
where in the last line we used the condition $\sum_{j=1}^{d} |\lambda_j| \leq 1$. Similarly we find that $\beta_{\pb}^{\rm SS} \geq -\gamma^0_{\pb}$ and hence the same range as in~\eqref{eq: bound_on_cosine_moment}. Both skewness measures thus are contained in the same bounded interval. 

\subsection{Dependence measures}\label{sec: depmeas}
We  consider two well-known types of correlation coefficients  proposed in the literature.
For $s,t \in \{1,\ldots,d\}$ we define
\begin{align*}
    \Sigma_s = \begin{pmatrix}
        \E \left[ \cos(\Theta_s) \right] \\
        \E \left[ \sin(\Theta_s) \right]
    \end{pmatrix}, 
    \quad
    \Sigma_{st} = \begin{pmatrix}
        \E \left[ \cos(\Theta_s) \cos(\Theta_t) \right] & 
        \E \left[ \sin(\Theta_s) \cos(\Theta_t) \right] \\
        \E \left[ \sin(\Theta_s) \cos(\Theta_t) \right] & 
        \E \left[ \sin(\Theta_s) \sin(\Theta_t) \right]
    \end{pmatrix}.
\end{align*}
Using this notation, we investigate the following two dependence measures and note that further can be found in \cite{jupp1989unified}:
\begin{equation*}
    \text{\cite{FL83}: } \rho_{\text{FL}} = \frac{\det(\Sigma_{st})}{\left( \det(\Sigma_{ss}) \det(\Sigma_{tt}) \right)^{1/2}},
    \quad
    \text{\cite{mackenzie1957estimation, STEPHENS1979}: } \rho_{\text{SM}} = \text{tr} \left[ \left( \Sigma_{st} \Sigma_{ts} \right)^{1/2} \right].
\end{equation*}
From the definition of trigonometric moments, it holds that $\Sigma_s = \left( \alpha_{\eb_s}^\intercal, \beta_{\eb_s}^\intercal \right)^\intercal,$ and through trigonometric identities it can be shown that 
\begin{align*}
    \Sigma_{st} = \frac{1}{2} \begin{pmatrix}
        \alpha_{\eb_s + \eb_t} + \alpha_{\eb_s - \eb_t} & 
        \beta_{\eb_s + \eb_t} + \beta_{\eb_s - \eb_t} \\
        \beta_{\eb_s + \eb_t} + \beta_{\eb_s - \eb_t} & 
        \alpha_{\eb_s - \eb_t} - \alpha_{\eb_s + \eb_t}
    \end{pmatrix}
\quad \rm{and} \quad 
    \Sigma_{ss} = \frac{1}{2} \begin{pmatrix}
        \alpha_{2\eb_s} + 1 & 
        \beta_{2\eb_s} \\
        \beta_{2\eb_s} & 
        1 - \alpha_{2\eb_s} 
    \end{pmatrix}.
\end{align*}

\subsection{Random variate generation}\label{sec: randomgen}

We will now show how one can generate random variates from our arctan-skewed model in a straightforward way. Theorem~\ref{thm: random_variate_generation} below provides  the simple algorithm and its proof is given in Section~\ref{sec: technical details}.

\begin{theorem} \label{thm: random_variate_generation}
    Let $f_0\in\mathcal{F}$. Consider the following simple algorithm to generate a random quantity $\Thetab\in\mathcal{T}_d$:
    \begin{enumerate}
        \item[Step 1.] Generate a random vector $\Xb = (X_1, \ldots, X_d)^\intercal$ from $f_0(\cdot - \mub)$;
        \item[Step 2.] Generate an independent random variable $U$ from the uniform distribution $U[0,1]$;
        \item[Step 3.] Set 
        \begin{equation*}
            \Thetab = \left\{ \begin{array}{cc}
                \Xb & \text{if } U \leq \frac{1}{2} \left( 1 + \frac{2}{\pi d} \sum_{i=1}^d \arctan \left( \frac{\lambda_i \sin (X_i - \mu_i)}{1 + \sin^2 (X_i - \mu_i)} \right) \right) \\
                2\mub - \Xb & \text{if } U > \frac{1}{2} \left( 1 + \frac{2}{\pi d} \sum_{i=1}^d \arctan \left( \frac{\lambda_i \sin (X_i - \mu_i)}{1 + \sin^2 (X_i - \mu_i)} \right) \right)
            \end{array}\right..
        \end{equation*}
    \end{enumerate}
    The resulting random vector $\Thetab$ then follows the $f_0$-arctan-skewed density~\eqref{eq: sum_model}.
\end{theorem}

 We note the simplicity and efficiency of the algorithm, which is well in line with the notoriously good random number generation algorithms accompanying the symmetry modulation approach, as mentioned in the Introduction.

\section{Inference \label{sec: inference}}

We first develop maximum likelihood estimation for the parameters of asymmetric models resulting from our arctan-skewing in Section~\ref{sec: MLE}. 
The associated (observed) Fisher Information matrix is given in Section~\ref{sec: FIM} of the Supplementary Material. 
Based on the Maximum Likelihood Estimators (MLEs), one can right away develop likelihood ratio tests for testing symmetry ($\lambdab = \zerob$) against arctan-skewed alternatives. 
We do not provide these tests here, as we propose instead efficient semi-parametric tests for symmetry in Section~\ref{sec: optimtest}. 
The performance of those tests is evaluated by means of Monte Carlo simulations in Section~\ref{sec: simulations} of the Supplementary Material. 
We consider there the setting $d=3$ for which we generate the symmetric data from the trivariate wrapped Cauchy (TWC) copula \citep{TWCC, TWCM} with marginal distributions chosen to be the wrapped Cauchy distribution.

\subsection{Maximum likelihood estimation \label{sec: MLE}}

We now develop methods to estimate the parameters of arctan-skewed asymmetric models. More concretely, we will perform maximum likelihood estimation. For notational simplicity, we adopt the notations

$$
M(\mu_k, \lambda_k; \theta_{jk}) = \frac{\partial}{\partial{\mu_k}}\left(1 + \frac{2}{\pi d} \sum_{i=1}^d \arctan \left( \frac{\lambda_i \sin (\theta_{ji} - \mu_{i})}{1 + \sin^2 (\theta_{ji} - \mu_{i})} \right)\right)=\frac{2\lambda_k}{\pi d}\frac{-\cos^3 (\theta_{jk} - \mu_k)}{\left( 1 + \sin^2 (\theta_{jk} - \mu_k) \right)^2+\lambda_k^2 \sin^2 (\theta_{jk} - \mu_k)}
$$
and
$$
L(\mu_k, \lambda_k; \theta_{jk}) = \frac{\partial}{\partial{\lambda_k}}\left(1 + \frac{2}{\pi d} \sum_{i=1}^d \arctan \left( \frac{\lambda_i \sin (\theta_{ji} - \mu_{i})}{1 + \sin^2 (\theta_{ji} - \mu_{i})} \right)\right)
= \frac{2}{\pi d}\frac{ \sin (\theta_{jk} - \mu_k) \left( 1 + \sin^2 (\theta_{jk} - \mu_k) \right)}
{\left( 1 +  \sin^2 (\theta_{jk} - \mu_k) \right)^2 + \lambda_k^2 \sin^2 (\theta_{jk} - \mu_k)},
$$
and note that 
\begin{equation*}
    M(\mu_k, 0; \theta_{jk})=0 
    \qquad \text{and} \qquad
    L(\mu_k, 0; \theta_{jk})=\frac{2}{\pi d}\frac{ \sin (\theta_{jk} - \mu_k)}
{1 +  \sin^2 (\theta_{jk} - \mu_k)}.
\end{equation*}
To obtain MLEs, we consider the following score functions for the parameters $\mub$ and $\lambdab$:

\begin{eqnarray*}
\frac{\partial}{\partial{\mu_k}} \log \left(\prod_{j=1}^n f (\thetab_j - \mub) \right)
&=&
\sum_{j=1}^n \frac{\frac{\partial}{\partial{\mu_k}}f_0(\thetab_j-\mub)}{f_0(\thetab_j-\mub)} 
+ \sum_{j=1}^n \frac{M(\mu_k, \lambda_k; \theta_{jk})}{ 1 + \frac{2}{\pi d} \sum_{i=1}^d \arctan \left( \frac{\lambda_i \sin (\theta_{ji} - \mu_{i})}{1 + \sin^2 (\theta_{ji} - \mu_{i})} \right) }
\end{eqnarray*}
and
$$\frac{\partial}{\partial{\lambda_k}} \log\left(\prod_{j=1}^n f (\thetab_j - \mub) \right)
=\sum_{j=1}^n \frac{L(\mu_k, \lambda_k; \theta_{jk})}
{1 + \frac{2}{\pi d} \sum_{i=1}^d \arctan \left( \frac{\lambda_i \sin (\theta_{ji} - \mu_{i})}{1 + \sin^2 (\theta_{ji} - \mu_{i})} \right)}.
$$
The MLEs $\hat{\mub}$ and $\hat{\lambdab}$ are obtained by setting these score equations equal to 0 and solving the resulting score equations numerically, along with other score equations that could arise through the existence of some other parameter $\varthetab$ that might be present in the base density $f_0$ (concentration, peakedness, shape, \ldots). Contrary to the sine-skewed models, we have no constraints on the skewness parameter, which we have observed in practice to simplify estimation.

\subsection{Optimal tests for symmetry}\label{sec: optimtest}

In this section we construct  tests for symmetry which are most efficient against our arctan-skewed alternatives and are moreover semi-parametric, {unlike purely parametric likelihood ratio tests}. Here we will propose procedures that are valid under a  broad class of symmetric densities. Moreover, we will consider both cases of tests about a specified and unspecified symmetry center. The backbone of our approach is the Le Cam theory of asymptotic experiments and we draw upon the methodology presented in~\cite{symmetry_test} (we also refer the reader to that paper for details about the Le Cam theory). {Of course, our procedure can also be adapted to optimal tests for symmetry in subvectors of $\lambdab$.}

In what follows, the symmetric densities $f_0$ need to satisfy slightly stronger conditions than those of $\mathcal{F}$, and we formulate those in the class of functions
\begin{align}\label{eq: def: H}
    \mathcal{H} := \left\{\begin{array}{ll}
& f_0(\thetab) > 0 \ \forall \thetab \in \mathcal{T}_d, \\
f_0: & f_0 \text{ unimodal in } \thetab \in \mathcal{T}_d \text{ with a mode at } \boldsymbol{0}, \\
& \thetab \mapsto f_{0}^{1/2}\left( \thetab\right) \text{ is } C^1 \forall \thetab \in \mathcal{T}_d \text{ with continuous derivative }
    \nabla_{\mub} f^{1/2}_{0}\left( \thetab - \mub \right). 
\end{array} \right\}
\end{align}
where $C^1$ denotes the class of continuously differentiable functions and it holds that
\begin{equation*}
    \nabla_{\mub} f^{1/2}_{0}\left( \thetab - \mub \right) 
    = \begin{pmatrix}
        \frac{\frac{\partial} {\partial\mu_1}  f_{0}\left( \thetab - \mub \right) } {2 f^{1/2}_{0}\left( \thetab - \mub \right) } \\
        \vdots \\
        \frac{\frac{\partial} {\partial\mu_d}  f_{0}\left( \thetab - \mub \right) } {2 f^{1/2}_{0}\left( \thetab - \mub \right) }
    \end{pmatrix}
    = \frac{1}{2} \begin{pmatrix}
         \phi^{f_0}_1(\thetab - \mub) \\
         \vdots \\
         \phi^{f_0}_d(\thetab - \mub)
    \end{pmatrix} f^{1/2}_{0}\left( \thetab - \mub \right),
\end{equation*}
where 
\begin{equation*} \label{eq: def: phi}
    \phi^{f_0}_j(\thetab - \mub) 
    = -\frac{\frac{\partial}{\partial \theta_j}f_0(\thetab -\mub)}{f_0(\thetab - \mub)}, \textrm{ for } j \in \{ 1, \ldots, d \}.
\end{equation*}
Most known models from the literature, such as the von Mises, wrapped Cauchy, cardioid, Sine, Cosine, bivariate wrapped Cauchy, or the trivariate
wrapped Cauchy with specified marginals \citep{TWCC, TWCM} satisfy the condition to belong to the intersection $\mathcal{F}\cap\mathcal{H}$. 

Let us denote by $P^{(n)}_{\mub, \lambdab; f_0}$, with $f_0 \in \mathcal{F} \cap \mathcal{H}$ as in \eqref{eq: def: F} and \eqref{eq: def: H}, the joint distribution of a sequence $\{ \thetab_i \}_{i=1}^n = \{(\theta_{i1}, \ldots, \theta_{id})^\intercal \}_{i=1}^n$ of independent and identically distributed hyper-toroidal random observations with density \eqref{eq: sum_model}. 
Any $f_0$ then induces a toroidal location-skewness model 
$$
P^{(n)}_{f_0}:=\left\{P^{(n)}_{\mub, \lambdab; f_0}: \mub \in \mathcal{T}_d, \lambdab \in \R^d\right\}.
$$
In order to derive tests for symmetry that are optimal in the Le Cam sense, we  establish the Uniform Local Asymptotic Normality (ULAN) property in the vicinity of symmetry, meaning at $\lambdab=\pmb0$,  of the parametric model $P^{(n)}_{f_0}$.

\begin{proposition}\label{prop: ULAN}
Let $f_0 \in \mathcal{F} \cap \mathcal{H}$. 
Then the arctan-skewed family $P^{(n)}_{f_0}$ is ULAN at $(\mub^\intercal,\zerob^\intercal)^\intercal$. More precisely, for any $\mu_i^{(n)}=\mu_i + O(n^{-1/2}), i=1,\ldots,d,$ and  for any bounded sequence $\taub^{(n)}:=(\taub^{(n)\intercal}_{\mu}, \taub^{(n)\intercal}_{\lambda})^\intercal \in\R^{2d}$ with $\taub^{(n)}_{\mu} = (\tau^{(n)}_1, \ldots, \tau^{(n)}_d)^\intercal \in \R^{d}$, $\taub^{(n)}_{\lambda} = (\tau^{(n)}_{d+1}, \ldots, \tau^{(n)}_{2d})^\intercal \in \R^{d}$ such that $\mu^{(n)}_i + n^{-1/2} \tau^{(n)}_i$, $i=1, \ldots, d,$ remains in $\mathcal{T}_1$, we have
\begin{align*} 
    \Lambda^{(n)} 
    & := \log \left( \frac{P^{(n)}_{\mub^{(n)} + n^{-1/2} \taub^{(n)}_{\mu}, n^{-1/2} \taub^{(n)}_{\lambda}; f_0} }{P^{(n)}_{\mub^{(n)}, \boldsymbol{0}; f_0}} \right) = \taub^{(n) \intercal} \Delta_{f_0}^{(n)}(\mub^{(n)}) - \frac{1}{2} \taub^{(n)\intercal} \Gamma_{f_0} \taub^{(n)} + o_P(1)
\end{align*}
as $n\to\infty$, and the central sequence
\begin{align*} \label{eq: def: central sequence}
    \Delta_{f_0}^{(n)}(\mub) 
    := & \frac{1}{\sqrt{n}}\sum_{i=1}^n \Delta_{f_0; i}^{(n)}(\mub) \\
    = & \frac{1}{\sqrt{n}}\sum_{i=1}^n \left(
        \phi^{f_0}_1(\thetab_{i} - \mub), \ \ldots, \
        \phi^{f_0}_d(\thetab_{i} - \mub), \
        \frac{2}{\pi d} \frac{\sin(\theta_{i1}-\mu_1)}{ 1+ \sin^2(\theta_{i1}-\mu_1)}, 
        \ \ldots, \
        \frac{2}{\pi d} \frac{\sin(\theta_{id}-\mu_d)}{ 1+ \sin^2(\theta_{id}-\mu_d)} 
    \right)^\intercal,
\end{align*}
for $\phi^{f_0}_j(\thetab-\mub)$ as in \eqref{eq: def: phi},
converges to $\mathcal{N}_{2d}(\zerob,\Gamma_{f_0})$ under $P^{(n)}_{\mub^{(n)}, \zerob; f_0}$ as $n\to\infty$, where the Fisher Information matrix $\Gamma_{f_0}$ is given by
\begin{equation*} \label{eq: def: Fisher Information}
    \Gamma_{f_0} = 
    \begin{pmatrix}
        I^{f_0}_{\mu_1\mu_1} & \hdots & I^{f_0}_{\mu_1\mu_d} & I^{f_0}_{\mu_1\lambda_1} & \hdots & I^{f_0}_{\mu_1\lambda_d} \\
        \vdots & \ddots & \vdots & \vdots & \ddots & \vdots \\
        I^{f_0}_{\mu_1\mu_d} & \hdots & I^{f_0}_{\mu_d\mu_d} & I^{f_0}_{\mu_d\lambda_1} & \hdots & I^{f_0}_{\mu_d\lambda_d} \\
        I^{f_0}_{\mu_1\lambda_1} & \hdots & I^{f_0}_{\mu_d\lambda_1} & I^{f_0}_{\lambda_1\lambda_1} & \hdots & I^{f_0}_{\lambda_1\lambda_d} \\
        \vdots & \ddots & \vdots & \vdots & \ddots & \vdots \\
        I^{f_0}_{\mu_1\lambda_d} & \hdots & I^{f_0}_{\mu_d\lambda_d} & I^{f_0}_{\lambda_1\lambda_d} & \hdots & I^{f_0}_{\lambda_d\lambda_d}
    \end{pmatrix}
\end{equation*}
with
    \begin{align} \label{eq: def: I}
        & I^{f_0}_{\mu_j\mu_k} = \int_{\mathcal{T}_d} \phi_j^{f_0}(\boldsymbol{\theta}) \phi_k^{f_0}(\boldsymbol{\theta}) f_0(\boldsymbol{\theta}) {\mathrm d}\boldsymbol{\theta}, \qquad  I^{f_0}_{\mu_j\lambda_j} = \frac{2}{\pi d} \int_{\mathcal{T}_d} \phi_j^{f_0}(\boldsymbol{\theta}) \frac{ \sin (\theta_{j})}
        {1 +  \sin^2 (\theta_{j})} f_0(\boldsymbol{\theta}) {\mathrm d}\boldsymbol{\theta}, \nonumber \\
        & I^{f_0}_{\lambda_j \lambda_k} = \frac{4}{\pi^2 d^2} \int_{\mathcal{T}_d} \frac{ \sin (\theta_{j})}
    {1 +  \sin^2 (\theta_{j})} \frac{ \sin (\theta_{k})}
    {1 +  \sin^2 (\theta_{k})}  f_0(\boldsymbol{\theta}) {\mathrm d}\boldsymbol{\theta}, \qquad I^{f_0}_{\mu_j\lambda_k} = 0, \text{ for } j\neq k, 
    \end{align}
where $j,k \in \{1,\ldots,d\}$.
\end{proposition}

The proof of this theorem goes along the same lines as the proof of ULAN for sine-skewing in~\cite{symmetry_test} with only minimal variations. The quadratic mean differentiability of  arctan-skewed densities $f_0$ can be obtained via Proposition 2.1 of~\cite{symmetry_test}; we leave the details to the reader. We also draw the reader's attention to the fact that the elements of the Fisher Information matrix $\Gamma_{f_0}$ are  obtained from the observed Fisher Information quantities at $\lambdab=\zerob$ given in Section~\ref{sec: FIM} through  calculations involving integration by parts.

\subsubsection{Optimal tests for symmetry about a specified center}\label{sec: testsymspec}
We firstly assume that the center of symmetry  is known and fixed to $\mub \in \mathcal{T}_d$.
Our goal is to derive semi-parametrically optimal tests for the following hypothesis testing problem:
\begin{equation*}\label{eq: hypotheses known median}
    \mathcal{H}^{(n)}_{0; \mu}:= \bigcup_{f_0 \in\mathcal{F} \cap \mathcal{H}} P^{(n)}_{\mub, \boldsymbol{0}; f_0}
    \quad \text{vs} \quad
    \mathcal{H}^{(n)}_{1; \mu}:= \bigcup_{\lambdab \in \R^d \setminus \{\boldsymbol{0} \} } \bigcup_{f_0 \in\mathcal{F} \cap \mathcal{H}} P^{(n)}_{\mub, \lambdab; f_0}.
\end{equation*}
Since $\mub$ is known, we focus on the parts of the central sequence and Fisher Information matrix that only depend on $\lambdab$. 
To this end, we define 
\begin{equation}\label{eq: def: central sequence and FI for known median}
    \Delta_{\lambda}^{(n)}(\mub) = \frac{1}{\sqrt{n}}\sum_{i=1}^n 
    \begin{pmatrix}
        \frac{2}{\pi d} \frac{\sin(\theta_{i1}-\mu_1)}{ 1+ \sin^2(\theta_{i1}-\mu_1)} \\
        \vdots \\
        \frac{2}{\pi d} \frac{\sin(\theta_{id}-\mu_d)}{ 1+ \sin^2(\theta_{id}-\mu_d)}
    \end{pmatrix} 
    \quad \text{and} \quad
    \Gamma_{f_0; \lambda} =  
    \begin{pmatrix}
        I^{f_0}_{\lambda_1\lambda_1} & \hdots & I^{f_0}_{\lambda_1\lambda_d} \\
        \vdots & \ddots & \vdots \\
        I^{f_0}_{\lambda_1\lambda_d} & \hdots & I^{f_0}_{\lambda_d\lambda_d}
    \end{pmatrix}
\end{equation}
where $I^{f_0}_{\lambda_i\lambda_j}$ for $i,j=1,\ldots,d$ is defined in \eqref{eq: def: I}. Note that the central sequence for skewness does not depend on $f_0$, meaning that the most efficient parametric test statistic $\left(\Delta_{\lambda}^{(n)}(\mub)\right)^\intercal \left({\Gamma}_{f_0;\lambda} \right)^{-1} \Delta_{\lambda}^{(n)}(\mub)$ can readily be turned into a semi-parametric one by estimating $\Gamma_{f_0;\lambda}$, and that any $f_0$ will lead to the same most efficient semi-parametric test, which we present and study in the following result whose proof is left to the reader (see the proof of Theorem 3.1 in \cite{symmetry_test}).

\begin{theorem} \label{thm: optimality specified median}
We define the semi-parametric test statistic as
\begin{equation*}
    Q^{\ast(n); \mub} = \left(\Delta_{\lambda}^{(n)}(\mub)\right)^\intercal \left(\hat{\Gamma}_{\lambda} \right)^{-1} \Delta_{\lambda}^{(n)}(\mub),
\end{equation*}
where 
\begin{equation*}\label{eq: Fisher Information specified estimate}
    \hat{\Gamma}_{\lambda} = 
    \begin{pmatrix}
        \hat{I}_{\lambda_1\lambda_1} & \hdots & \hat{I}_{\lambda_1\lambda_d} \\
        \vdots & \ddots & \vdots \\
        \hat{I}_{\lambda_1\lambda_d} & \hdots & \hat{I}_{\lambda_d\lambda_d}
    \end{pmatrix},
    \quad \hat{I}_{\lambda_j\lambda_k} = \frac{4}{n \pi^2 d^2} \sum_{i=1}^n \frac{ \sin (\theta_{ij} - \mu_j)}
    {1 +  \sin^2 (\theta_{ij} - \mu_j)} \frac{ \sin (\theta_{ik} - \mu_k)}
    {1 +  \sin^2 (\theta_{ik} - \mu_k)} \text{  for  } j,k=1,\ldots,d,
\end{equation*}
is the empirical estimate of $\Gamma_{f_0; \lambda}$ defined in \eqref{eq: def: central sequence and FI for known median}. 
Let $\phi^{\ast (n); \mu}$ denote the test that rejects $\mathcal{H}^{(n)}_{0; \mu}$ against the alternative $\mathcal{H}^{(n)}_{1; \mu}$ at asymptotic level $\alpha$ whenever the statistic $Q^{\ast(n); \mub} > \chi^2_{d; \alpha}$, with $\chi^2_{d; \alpha}$ the $\alpha$-upper quantile of a chi-square distribution with $d$ degrees of freedom.
The following hold for $f_0 \in \mathcal{F} \cap \mathcal{H}$:
    \begin{enumerate}
        \item[(i)] Under $\mathcal{H}^{(n)}_{0; \mu}$, we have that
        \begin{equation*} \label{eq: asymptotic_result_specified}
            Q^{\ast(n); \mub}   \xrightarrow{\mathcal{D}} 
            \chi^2_d
        \end{equation*}
    as $n \rightarrow \infty$, so that the test $\phi^{\ast (n); \mu}$ has asymptotic level $\alpha$ under the null.
        \item[(ii)] Under $P^{(n)}_{\mub, n^{-1/2} \taub^{(n)}_\lambda; f_0}$ for $f_0 \in \mathcal{F} \cap \mathcal{H}$ and $\taub^{(n)}_\lambda$ as defined in Proposition~\ref{prop: ULAN}, it holds that
    \begin{equation*}\label{eq: dist_under_alternative d dimensions}
        Q^{\ast(n); \mub} \xrightarrow{\mathcal{D}} 
        \chi^2_d(\taub^{\intercal}_\lambda \Gamma_{f_0; \lambda} \taub_\lambda)
    \end{equation*}
     as $n \rightarrow \infty$, where $\chi^2_\nu(\kappa)$ denotes a non-central $\chi^2$ distribution with $\nu$ degrees of freedom and non-centrality parameter $\kappa$, with $\taub_\lambda = \lim_{n\rightarrow\infty} \taub^{(n)}_\lambda$.

     \item[(iii)] The test $\phi^{\ast (n); \mu}$ is universally (in $f_0$) locally and asymptotically maximin at level $\alpha$ when testing $\mathcal{H}^{(n)}_{0; \mu}$ against $\mathcal{H}^{(n)}_{1; \mu}$.
    \end{enumerate}
\end{theorem}

The universal optimality under Theorem~\ref{thm: optimality specified median}(iii) follows from the aforementioned fact that the central sequence for skewness does not depend on $f_0$. Thanks to Theorem~\ref{thm: optimality specified median}(ii), one can even express the power of our test in a concise formula.

\subsubsection{Optimal tests for symmetry about an unspecified center}\label{sec: testsymunspec}

We now consider the case of an unspecified symmetry center.
In formal terms we consider the hypothesis testing problem
\begin{equation*}\label{eq: hypotheses unknown median}
    \mathcal{H}_0^{(n)}:= \bigcup_{\mub \in \mathcal{T}_d} \bigcup_{g_0 \in\mathcal{F} \cap \mathcal{H}} P^{(n)}_{\mub, \zerob; g_0}
    \quad \text{vs} \quad
    \mathcal{H}_1^{(n)}:= \bigcup_{\lambdab \in \R^d \setminus \{\boldsymbol{0} \} } \bigcup_{\mub \in \mathcal{T}_d} \bigcup_{g_0 \in\mathcal{F} \cap \mathcal{H}} P^{(n)}_{\mub, \lambdab; g_0}.
\end{equation*}
The principal change compared to the previous section lies in the fact that we now need to estimate the nuisance parameter $\mub$ by some estimator $\hat{\mub}^{(n)}$. To this end as well as for other developments, we adopt the following assumption.

\begin{assumption} \label{assumption}
\begin{enumerate}[(i)]
    \item The estimators $\hat{\mub}^{(n)} = \left( \hat{\mu}^{(n)}_1, \ldots, \hat{\mu}^{(n)}_d \right)^\intercal$ for $\mub$ are $\sqrt{n}$-consistent and locally asymptotic discrete.

    \item The mapping $\thetab \mapsto f_0(\thetab - \mub)$ is $\mathcal{C}^2$ almost everywhere on $\mathcal{T}_d$.
\end{enumerate}
\end{assumption}
Assumption~\ref{assumption}(i) is needed to ensure that we can use Lemma 4.4 of~\cite{kreiss_adaptive_1987} which is a crucial ingredient to replace a nuisance parameter $\mub$ with an estimator and still get asymptotic results. We refer the reader to Section~4 of \cite{symmetry_test} for more details, as well as a definition and discussion of the concept of locally asymptotic discrete estimators. Assumption~\ref{assumption}(i)  is a purely technical requirement, and for a fixed sample size, any estimator can be considered part of a locally asymptotically discrete sequence.

Constructing optimal tests for symmetry under an unspecified center requires many more steps compared to the specified center scenario. Establishing optimal tests requires working with two distinct symmetric densities $f_0, g_0\in\mathcal{F}\cap\mathcal{H}$, $f_0$ being the assumed one and $g_0$ the true unknown one (see below for more details), and we need to define the following:
\begin{equation*}
\Delta_{f_0; \mu}^{(n)}(\mub) =
    \frac{1}{\sqrt{n}}\sum_{i=1}^n \begin{pmatrix}
        \phi^{f_0}_1(\thetab_{i} - \mub), & \ldots, &
        \phi^{f_0}_d(\thetab_{i} - \mub) 
    \end{pmatrix}^\intercal,
\end{equation*}
\begin{align*}
    C^{f_0; g_0}_{\mu; \mu}
    := & \cov_{g_0} \left(\Delta_{f_0; \mu}^{(n)}(\mub), \Delta_{g_0; \mu}^{(n)}(\mub) \right) 
    = \begin{pmatrix}
        I^{f_0;g_0}_{\mu_1\mu_1} & \cdots & I^{f_0;g_0}_{\mu_1\mu_d} \\
        \vdots & \ddots & \vdots \\
        I^{f_0;g_0}_{\mu_d\mu_1} & \cdots & I^{f_0;g_0}_{\mu_d\mu_d}
    \end{pmatrix}, \\
    C^{g_0}_{\mu; \lambda}
    := & \cov_{g_0} \left( \Delta_{\lambda}^{(n)}(\mub), \Delta_{g_0; \mu}^{(n)}(\mub) \right) 
    = \begin{pmatrix}
        I^{g_0}_{\mu_1\lambda_1} & \cdots & 0 \\
        \vdots & \ddots & \vdots \\
        0 & \cdots & I^{g_0}_{\mu_d\lambda_d}
    \end{pmatrix}
\end{align*}
where 
\begin{align}\label{eq: param_cov_entries}
    I^{f_0;g_0}_{\mu_j\mu_k} = \int_{\mathcal{T}_d} \phi_j^{f_0}(\thetab - \mub) \phi_k^{g_0}(\thetab - \mub) g_0(\thetab - \mub) {\mathrm d} \thetab, \quad I^{g_0}_{\mu_j\lambda_j} = \frac{2}{\pi d} \int_{\mathcal{T}_d} \frac{\sin(\theta_{j}-\mu_j)}{ 1+ \sin^2(\theta_{j}-\mu_j)} \phi_j^{g_0}(\thetab - \mub) g_0(\thetab - \mub) {\mathrm d} \thetab
\end{align}
for $j,k\in\{1,\ldots,d\}$ as well as their estimates
\begin{equation*}
    \widehat{C}^{f_0}_{\mu; \mu} 
    = \begin{pmatrix}
        \hat{I}^{f_0}_{\mu_1\mu_1} & \cdots & \hat{I}^{f_0}_{\mu_1\mu_d} \\
        \vdots & \ddots & \vdots \\
        \hat{I}^{f_0}_{\mu_1\mu_d} & \cdots & \hat{I}^{f_0}_{\mu_d\mu_d}
    \end{pmatrix},
    \quad
    \widehat{C}_{\mu; \lambda}
    = \begin{pmatrix}
        \hat{I}_{\mu_1\lambda_1} & \cdots & 0 \\
        \vdots & \ddots & \vdots \\
        0 & \cdots & \hat{I}_{\mu_d\lambda_d}
    \end{pmatrix},
\end{equation*}
where
\begin{equation}\label{eq: semi_param_cov_entries}
    \hat{I}^{f_0}_{\mu_j\mu_k} =  \frac{1}{n}\sum_{i=1}^n \frac{\partial}{\partial\theta_k}\phi_j^{f_0}(\thetab_{i} - \hat{\mub}^{(n)}),
    \quad
    \hat{I}_{\mu_j\lambda_j} = \frac{1}{n \pi d \sqrt{2}} \sum_{i=1}^n \log \left( \frac{\sec ^2 \left( \left( \theta_{ji} - \hat{\mu}^{(n)}_j\right)/{2} \right) -2 \sqrt{2}+2}{\sec ^2 \left( \left(\theta_{ji} - \hat{\mu}^{(n)}_j \right)/{2} \right)+2 \sqrt{2}+2} \right).
\end{equation}
The quantities~\eqref{eq: semi_param_cov_entries} are indeed estimates for~\eqref{eq: param_cov_entries} thanks to the law of large numbers and integration by parts under the measure induced by $g_0$. For the link between $\hat{I}_{\mu_j\lambda_j}$ and ${I}^{g_0}_{\mu_j\mu_k}$ we remind the equality~\eqref{eq: integral of weird sines}.

Equipped with these notations, we can develop our tests and start by optimal $f_0$-parametric tests which we later render valid under $g_0$. We first note that the Fisher Information matrix is not block-diagonal, meaning that asymptotically there is no independence under the null between the central sequences for skewness $\Delta_{\lambda}^{(n)}(\mub)$ and location $\Delta_{f_0;\mu}^{(n)}(\mub)$. Consequently, any variation in $\mub$ (as occurs when estimating it) induces an impact on the central sequence for skewness, which implies that we need to work with the so-called $f_0$-efficient central sequence
\begin{align*}
     \Delta_{f_0;\lambda}^{(n)\ast} (\mub)
     & = \Delta_{\lambda}^{(n)}(\mub) - \cov_{f_0}\left( \Delta_{\lambda}^{(n)}(\mub), \Delta_{f_0;\mu}^{(n)}(\mub) \right) \var_{f_0} \left( \Delta_{f_0;\mu}^{(n)} (\mub)\right)^{-1} \Delta_{f_0;\mu}^{(n)}(\mub),  
\end{align*}
obtained through a Gram-Schmidt orthogonalization procedure ensuring orthogonality of $\Delta_{f_0;\lambda}^{(n)\ast} (\mub)$
with respect to the location central sequence $\Delta_{f_0;\mu}^{(n)}(\mub)$. Under $f_0$, one can now replace $\mub$ with an estimator $\hat{\mub}^{(n)}$ satisfying Assumption~\ref{assumption} without altering the asymptotic distribution of $\Delta_{f_0;\lambda}^{(n)\ast} (\mub)$. The test statistic based on $\Delta_{f_0;\lambda}^{(n)\ast} (\hat{\mub}^{(n)})$, which we do not write out for the sake of space, will be optimal under $f_0$ but not valid under any other density $g_0$. Hence, whenever the true data density $g_0$ differs from the assumed density $f_0$, the parametric test  will not work, which is why we now focus on rendering it semi-parametric. To this end, we need to modify directly the $f_0$-efficient central sequence $\Delta_{f_0;\lambda}^{(n)\ast} (\mub)$. It is easy to see that $\cov_{g_0} \left(\Delta_{f_0;\lambda}^{(n)\ast} (\mub), \Delta_{g_0; \mu}^{(n)}(\mub) \right)\neq\zerob$, hence orthogonalization needs to occur under $g_0$, not under $f_0$. This leads to the new quantity
\begin{align}\label{eq: param test unspecified mu f,g}
    \Tilde{\Delta}_{f_0;g_0;\lambda}^{(n)} (\mub)
     & = \Delta_{\lambda}^{(n)}(\mub) - C^{g_0}_{\mu; \lambda} \left( C^{f_0; g_0}_{\mu; \mu} \right) ^{-1} \Delta_{f_0; \mu}^{(n)}(\mub),
\end{align}
which is now orthogonal to $\Delta_{g_0; \mu}^{(n)}(\mub)$ under the null hypothesis induced by $g_0$. The expression~$ \Tilde{\Delta}_{f_0;g_0;\lambda}^{(n)} (\mub)$ is thus the right starting point for building the desired semi-parametric tests. Replacing $\mub$ by an adequate estimator $\hat\mub^{(n)}$ and estimating the matrices by the estimates mentioned before yields the efficient central sequence
\begin{align}\label{eq: semi-param test unspecified mu f,g}
    \Tilde{\Delta}_{f_0; \lambda}^{(n)\ast} (\hat{\mub}^{(n)})
    = \Delta_{\lambda}^{(n)}(\hat{\mub}^{(n)}) - \widehat{C}_{\mu; \lambda}
    \left( \widehat{C}^{f_0}_{\mu; \mu} \right)^{-1} \Delta_{f_0; \mu}^{(n)}(\hat{\mub}^{(n)}).
\end{align}
In Proposition~\ref{prop: estimate quantities} below we show that the asymptotic behavior of $\Tilde{\Delta}_{f_0; \lambda}^{(n)\ast} (\hat{\mub}^{(n)})$ can be derived from that of $\Tilde{\Delta}_{f_0;g_0;\lambda}^{(n)} (\mub)$. This naturally brings us to build a most efficient test statistic on a quadratic product involving $\Tilde{\Delta}_{f_0; \lambda}^{(n)\ast} (\hat{\mub}^{(n)})$ and a quantity that estimates the variance of $ \Tilde{\Delta}_{f_0;g_0;\lambda}^{(n)} (\mub)$ under $g_0$, denoted by $V_{g_0}^{f_0} \left( \mub \right)$. Such an  estimate is given by
\begin{equation} \label{eq: variance estimate}
    \hat{V}_{f_0}\left(\hat{\mub}^{(n)} \right) 
    = \frac{1}{n} \sum_{i=1}^n \tilde{\Delta}^{(n)\ast}_{f_0; \lambda; i} \left(\hat{\mub}^{(n)} \right) \tilde{\Delta}^{(n) \ast}_{f_0; \lambda; i} \left( \hat{\mub}^{(n)} \right)^\intercal 
\end{equation}
for $\Tilde{\Delta}_{f_0; \lambda; i}^{(n)\ast} (\hat{\mub}^{(n)})$ defined via
\begin{equation*}
    \Tilde{\Delta}_{f_0; \lambda}^{(n)\ast} (\hat{\mub}^{(n)})
    = \frac{1}{\sqrt{n}} \sum_{i=1}^n \Tilde{\Delta}_{f_0; \lambda; i}^{(n)\ast} (\hat{\mub}^{(n)}).
\end{equation*}

\begin{proposition} \label{prop: estimate quantities}
    Suppose that $f_0, g_0 \in \mathcal{F} \cap \mathcal{H}$ and that Assumption \ref{assumption} holds.
    Then, for $j,k \in \{ 1,\ldots,d \}$, it holds that
    \begin{enumerate}
        \item[(i)] $
        \!
        \begin{aligned}[t]
        I^{g_0}_{\mu_j\lambda_j} - \hat{I}_{\mu_j\lambda_j} = o_P(1),
        \end{aligned}$ 
        as defined in \eqref{eq: param_cov_entries} and \eqref{eq: semi_param_cov_entries};
        
    \item[(ii)] $
        \!
        \begin{aligned}[t]
        I^{f_0;g_0}_{\mu_j\mu_k} - \hat{I}^{f_0}_{\mu_j\mu_k} = o_P(1), 
        \end{aligned}$
        as defined in \eqref{eq: param_cov_entries} and \eqref{eq: semi_param_cov_entries};
        
        \item[(iii)] $
        \!
        \begin{aligned}[t]
        \Tilde{\Delta}_{f_0; \lambda}^{(n)\ast} (\hat{\mub}^{(n)}) - \Tilde{\Delta}_{f_0;g_0;\lambda}^{(n)} (\mub) = o_P(1) \one_{d \times 1}, 
        \end{aligned}$
        as defined in \eqref{eq: param test unspecified mu f,g} and \eqref{eq: semi-param test unspecified mu f,g};

        \item[(iv)] $
        \!
        \begin{aligned}[t]
        \hat{V}_{f_0} \left(\hat{\mub}^{(n)} \right)  - V^{f_0}_{g_0}\left(\mub \right) = o_P(1)\one_{d \times d}, 
        \end{aligned}$
        as defined in \eqref{eq: variance estimate},
    \end{enumerate}
as $n\rightarrow\infty$ under $P^{(n)}_{\mub, \boldsymbol{0}; g_0}$, where $\one_{n\times m}$ is the $n\times m$ matrix of ones.
\end{proposition}
Elements of this proof are  given in Section~\ref{sec: technical details}. 
Proposition~\ref{prop: estimate quantities} is crucial in establishing the final semi-parametric test statistic and its asymptotic properties, as done in the  following result.

\begin{theorem} \label{thm: optimality unspecified median}
We define the test statistic 
\begin{equation*} \label{eq: semi-parametric test unspecified}
    Q_{f_0}^{\ast (n)} := \left(\Tilde{\Delta}_{f_0; \lambda}^{(n)\ast} (\hat{\mub}^{(n)}) \right)^\intercal \left(\hat{V}_{f_0} (\hat{\mub}^{(n)} ) \right)^{-1} \Tilde{\Delta}_{f_0; \lambda}^{(n)\ast} (\hat{\mub}^{(n)})
\end{equation*}
and let $\phi^{\ast (n)}_{f_0}$ denote the test that rejects $\mathcal{H}^{(n)}_{0}$ at asymptotic level $\alpha$ whenever the statistic $Q^{\ast(n)}_{f_0} > \chi^2_{d; \alpha}$.
Considering $f_0, g_0 \in \mathcal{F} \cap \mathcal{H}$ and working under Assumption \ref{assumption}, the following hold:
    \begin{enumerate}
        \item[(i)] Under $\mathcal{H}^{(n)}_{0}$,
        \begin{equation*}
            Q_{f_0}^{\ast(n)}   \xrightarrow{\mathcal{D}} \chi^2_d
        \end{equation*}
        as $n \rightarrow \infty$, so that the test $\phi^{\ast (n)}_{f_0}$ has asymptotic level $\alpha$ under the same hypothesis.

        \item[(ii)] Under $\cup_{\mub \ \in \mathcal{T}_d}P^{(n)}_{\mub, n^{-1/2} \boldsymbol{\tau}^{(n)}_\lambda; g_0}$,
\begin{equation*}\label{eq: dist_under_alternative unspecified}
    Q_{f_0}^{\ast (n)} \xrightarrow{\mathcal{D}} 
    \chi^2_d \Bigl( \boldsymbol{\tau}^\intercal_\lambda C_{g_0}^{f_0}(\mub) V_{g_0}^{f_0}(\mub)^{-1} C_{g_0}^{f_0}(\mub) \boldsymbol{\tau}_\lambda \Bigr)
\end{equation*}
as $n \rightarrow \infty$ with $\boldsymbol{\tau}_\lambda = \lim_{n\rightarrow\infty} \boldsymbol{\tau}^{(n)}_\lambda$
and $C_{g_0}^{f_0}(\mub) = \cov_{g_0}\left( \Tilde{\Delta}^{(n)}_{f_0;g_0;\lambda} (\mub), \Delta^{(n)}_{\lambda} (\mub) \right)$.

     \item[(iii)] The test $\phi^{\ast (n)}_{f_0}$ is locally and asymptotically maximin at level $\alpha$ when testing $\mathcal{H}^{(n)}_0$ against $\cup_{\mu \ \in \mathcal{T}_d}P^{(n)}_{\mub, n^{-1/2} \boldsymbol{\tau}^{(n)}_\lambda; f_0}$.
    \end{enumerate}
\end{theorem}

The proof of this theorem is left to the reader (see the proof of Theorem 4.3 of~\cite{symmetry_test}). We obtain a different test $\phi^{\ast (n)}_{f_0}$ for each density $f_0$ it is based on, and  each test is valid under the entire null hypothesis $\mathcal{H}^{(n)}_{0}$. Thanks to our result on Fisher Information singularity for arctan-skewing (see Theorem~\ref{thm: singularity}), we know that we can use any density $f_0$ from the literature without running into any issues, which is one of the main shortcomings of sine-skewing where optimal tests should not be based on $f_0$ being the von Mises or the Cosine density, for instance (see~\cite{symmetry_test} and~\cite{schutte_singularity_2026}). This is a further advantage of our arctan-skewing compared to sine-skewing.

\section{Real data analysis}\label{sec: real}
In this section, we compare  our arctan-skewing scheme to sine-skewing in terms of fitting real data. The chosen dataset includes data from the protein folding problem, which involves predicting the three-dimensional structure of a protein using the sequence of amino acids that make up the protein.
This is an important and challenging problem with implications to vaccine and medicine development. Proteins are made up of amino acids, which in turn are made up of the backbone and the sidechains.
In this analysis, we only focus on the backbone but one could also consider the angles formed in the sidechains.
The backbone consists of two flexible chemical bonds, NH-C$\alpha$ and C$\alpha$-CO, where C$\alpha$ denotes the central Carbon atom.
The angle that can be found at the beginning of each residue is the NH-C$\alpha$ torsion angle, denoted by $\phi$ and at the end is the C$\alpha$-CO torsion angle, denoted by $\psi$.
These angles need to be studied as specific combinations of them allow the favourable hydrogen bonding patterns, while others can `result in clashes within the backbone or between adjacent sidechains' \citep{jacobsen_introduction_2023}. The dihedral angles $\phi,\psi$ form data on the two-dimensional torus and \cite{symmetry_test} showed that such data are not symmetric, hence are best modelled with asymmetric toroidal distributions. 

For the present data analysis, we consider position 55 at 2000 randomly selected times in the molecular dynamic trajectory of the SARS-CoV-2 spike domain from \cite{genna_sars-cov-2_2020}.
The position occurs in $\alpha$-helix throughout the trajectory. 
Since we only use data from an $\alpha$-helix, the data used have a unique mode and need not be modelled by a mixture model. 
Applying our test $\phi^{\ast (n)}_{f_0}$ using the Sine distribution as $f_0$, with various combinations of $\kappa_1$, $\kappa_2$ and $\rho$, rejects the null hypothesis of symmetry at the 5\% level of significance.
Applying the same test with the Cosine model gives a p-value of around 0.06-0.07 for different values of the parameters, thus we can reject the null hypothesis at the 10\% level but not at 5\%.
The bwC model is not used as the data is likely too concentrated for the heavy tails of the bwC distribution.
Our evidence for asymmetry  confirms the conclusion of \cite{symmetry_test} when using 2 dihedral angles instead of 3.
Thus, we compute the fit of skewed models on our data, and more specifically
we compare the arctan-skewed and sine-skewed versions of the bwC, Sine and Cosine models in Table~\ref{tab: proteins}. 
We provide the maximum likelihood estimates of all parameters along with the value of the log-likelihood. 
Since all models have the same number of parameters, there is no need to give the Akaike or Bayesian Information Criteria and we compare the fitting capacity of the different models in terms of the log-likelihood: the higher the better. 
We see that for each symmetric density, its arctan-skewed version provides a better fit than the sine-skewed version. 
We observe that for the bwC and the Sine models, the sine-skewed version reaches the limit of the parameter space for the skewness parameter, meaning that $|\lambda_1|+|\lambda_2|=1$. 
The overall best fit is given by the arctan-skewed Sine model.
The skewness for each component corresponding to this model is given by $\mathcal{S}_1 = 0.36$ and $\mathcal{S}_2 = -0.46$.
The dependence between $\phi$ and $\psi$ is given by $\rho_{\text{FL}} = -0.83$, where we focus on the signed correlation coefficient $\rho_{\text{FL}}$. We note that both the skewness direction as well as the negative correlation correspond to what we observe visually, see Figure~\ref{fig: proteins} which also illustrates the good fit of the arctan-skewed Sine model. 
We notice the bad fit provided by any skewed version of the bwC model. 
This is due to the fact that the bwC is not suitable for light tails, as already mentioned.

\begin{table}[ht]
\centering
\begin{tabular}{llrrrrrrrr}
  \hline
 Skewing & Model & log-likelihood & $\hat{\mu}_1$ & $\hat{\mu}_2$ & $\hat{\kappa}_1$ & $\hat{\kappa}_2$ & $\hat{\rho}$ & $\hat{\lambda}_1$ & $\hat{\lambda}_2$ \\ 
  \hline
\multirow{3}{*}{Arctan} & bwC & -0.52 & 1.98 & 2.77 & 0.88 & 0.85 & -0.33 & -4.40 & 0.96 \\
& \textbf{Sine} & \textbf{719.39} & \textbf{1.85} & \textbf{2.93} & \textbf{33.49} & \textbf{20.80} & \textbf{-15.83} & \textbf{8.15} & \textbf{-7.51} \\ 
  & Cosine & 712.34 & 1.96 & 2.85 & 43.81 & 32.11 & -11.78 & -6.50 & -6.93 \\ 
  \hline
  \multirow{3}{*}{Sine} & bwC & -746.08 & 1.82 & -3.14 & 0.85 & 0.70 & -0.50 & -0.00 & -1.00 \\ 
  & Sine & 706.33 & 1.96 & 2.80 & 34.83 & 22.19 & -13.93 & -0.88 & -0.12 \\ 
  & Cosine & 706.06 & 1.96 & 2.80 & 48.80 & 36.17 & -14.37 & -0.77 & -0.13 \\ 
   \hline
\end{tabular}
\caption{Maximum likelihood estimates of the parameters of the bivariate wrapped Cauchy, Sine and Cosine models combined with arctan- or sine-skewing mechanisms, along with the resulting log-likelihood values. In boldface we indicate the best-fitting model.}
\label{tab: proteins}
\end{table}

\begin{figure}[htbp]
    \centering

    \begin{subfigure}{0.48\textwidth}
        \centering
        \includegraphics[height=5cm, width=\textwidth, trim=0cm 1cm 0cm 2cm, clip]{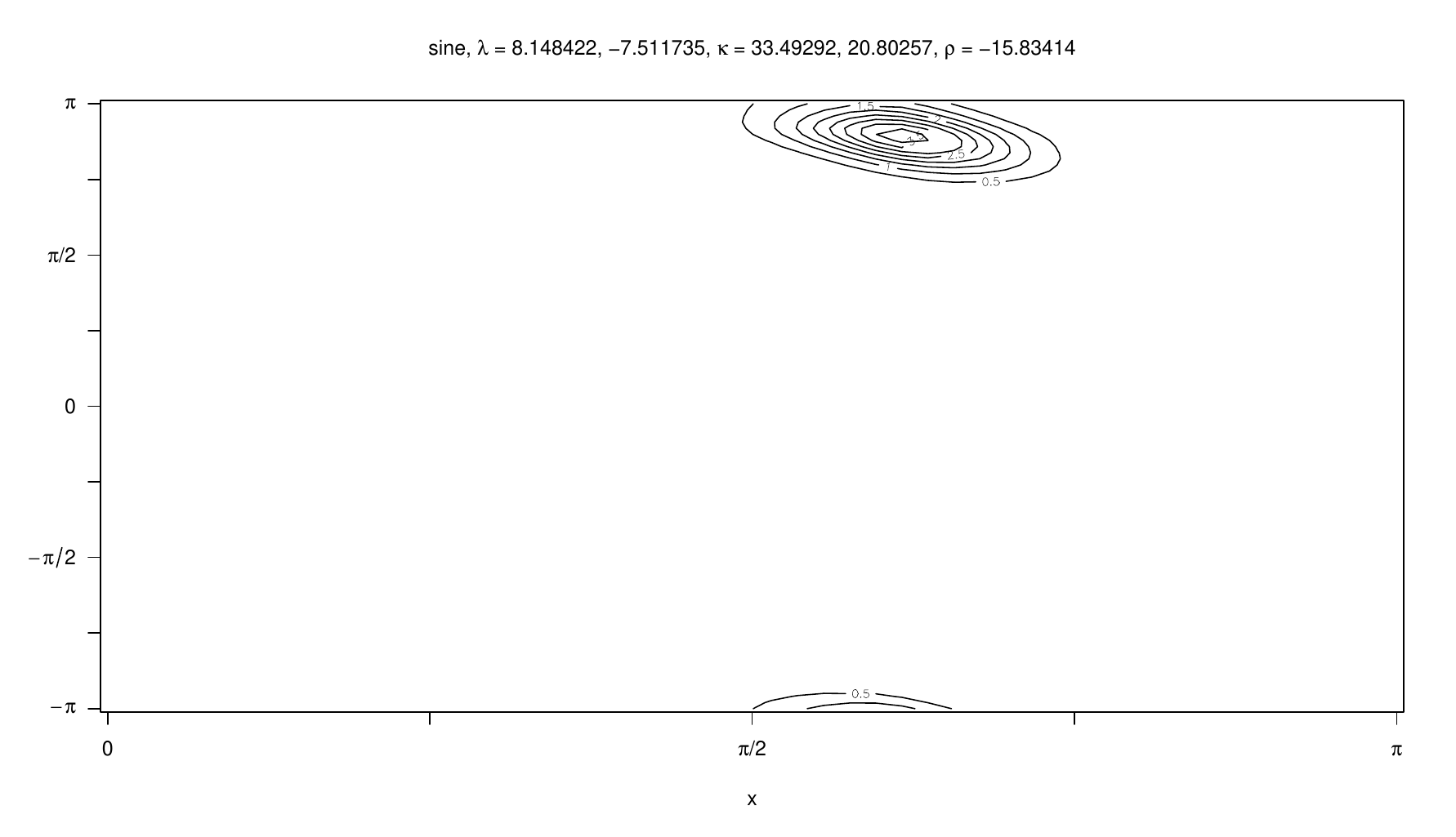}
    \end{subfigure}
    \begin{subfigure}{0.48\textwidth}
        \centering
        \includegraphics[height=5cm, width=\textwidth, trim=0cm 1cm 0cm 2cm, clip]{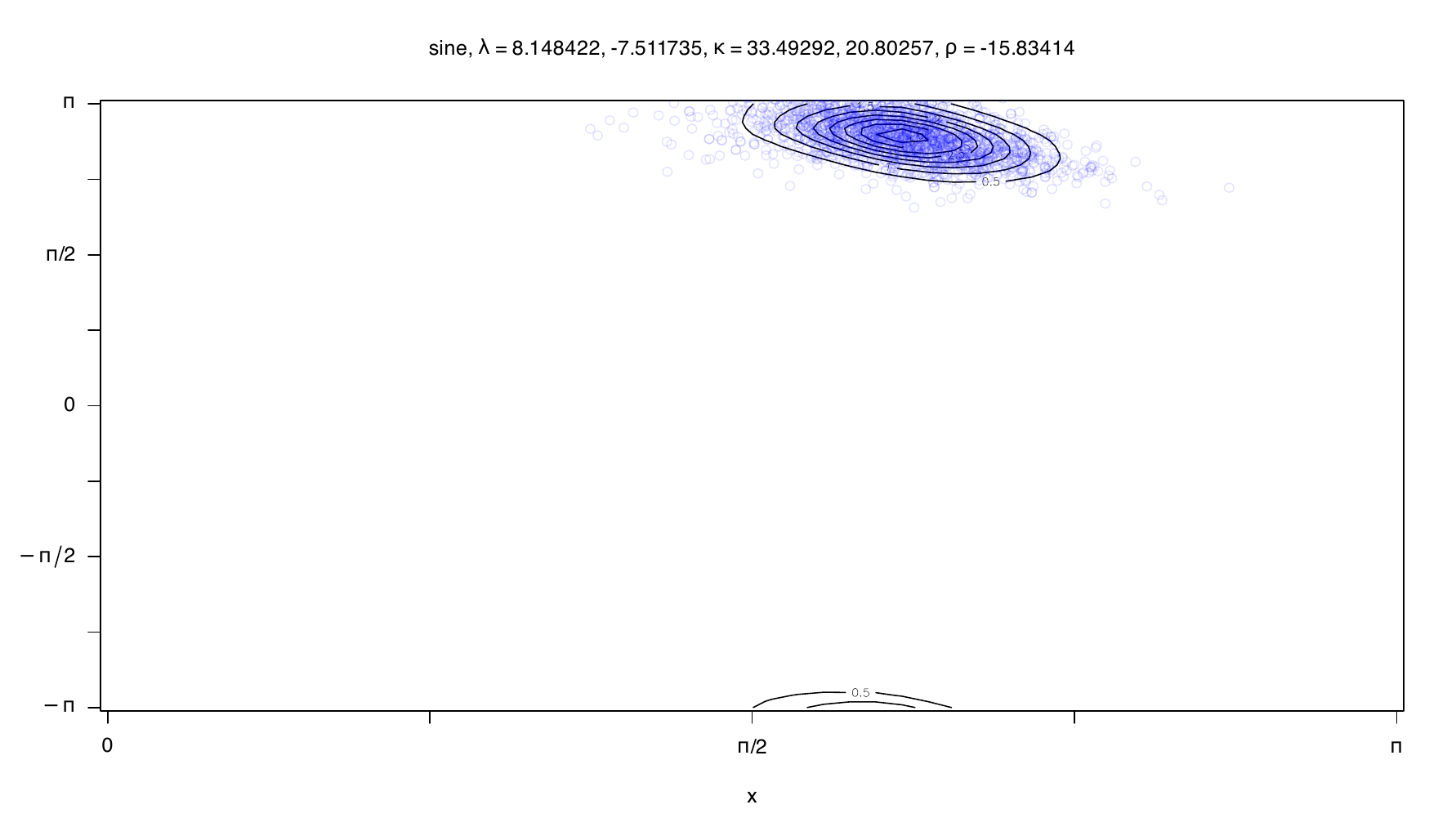}
    \end{subfigure}

    \caption{Contour plots of the best-fitting arctan-skewed Sine density with parameters $\hat{\mu}_1 = 1.85$, $\hat{\mu}_2 = 2.93$, $\hat{\kappa}_1 = 33.49$, $\hat{\kappa}_2 = 20.80$, $\hat{\rho} = -15.83$ and skewness parameters $\hat{\lambda}_1 = 8.15$   $\hat{\lambda}_2 = -7.51$.
    Plots without (left) and with (right) the data points are provided.}
    \label{fig: proteins}
\end{figure}

\section{Discussion}\label{sec: final}
In this paper, we have proposed the arctan-skewing scheme, an alternative to sine-skewing that retains the latter’s advantages while overcoming its main shortcomings, namely the restriction imposed on the skewness parameter and the singularity of the Fisher Information matrix in the vicinity of symmetry when sine-skewing is combined with some of the most popular symmetric densities $f_0$ from the literature like the circular von Mises or bivariate Cosine models. Removing the information singularity issue also entails the advantage that the optimal tests for symmetry we have built in this paper can be based on any density $f_0$ from the literature, unlike for the tests built in \cite{symmetry_test} which are optimal against sine-skewing.
The real data analysis on protein data moreover revealed  a superior fit of arctan-skewing over sine-skewing. Our arctan-skewing scheme therefore provides a compelling alternative to sine-skewing.

\section{Technical details} \label{sec: technical details}

\begin{proof}[\rm\textbf{Proof of Theorem~\ref{thm: singularity}.}]
The Fisher Information Matrix $I_{f_0}$ for any symmetric distribution $f_0$ subject to the arctan-skewing is given by
\begin{equation*}
    I_{f_0} = \E \left[S_{f_0} \left(S_{f_0}\right)^\intercal \right],
\end{equation*}
with $S_{f_0}$ the notation for the score function evaluated in the vicinity of symmetry. The matrix
$I_{f_0}$ is singular if and only if the components of $S_{f_0}$ are linearly dependent.
Immediate calculations yield that the score function of \eqref{eq: sum_model} in the vicinity of symmetry corresponds to
\begin{equation} \label{eq: score}
    S_{f_0} 
    =  \left( 
         -\frac{\frac{\partial}{\partial \theta_1}f_0(\thetab-\mub)}{f_0(\thetab-\mub)},
        \dots,
        -\frac{\frac{\partial}{\partial \theta_d}f_0(\thetab-\mub)}{f_0(\thetab-\mub)},
        \frac{2}{\pi d} \frac{\sin(\theta_{1}-\mu_1)}{ 1+ \sin^2(\theta_{1}-\mu_1)},
        \dots,
        \frac{2}{\pi d} \frac{\sin(\theta_{d}-\mu_d)}{ 1 + \sin^2(\theta_{d}-\mu_d)} 
    \right)^\intercal.
\end{equation}
For the rest of the proof, we set $\mub = \zerob$ without loss of generality, for notational convenience.
Therefore, the Fisher Information matrix is singular if and only if there exist real coefficients 
$\alpha_1,\ldots,\alpha_d,\beta_1,\ldots,\beta_d$ not all equal to 0 satisfying 
\begin{equation*}
    - \sum_{i=1}^d \alpha_i \frac{\frac{\partial}{\partial \theta_i}f_0(\thetab)}{f_0(\thetab)} + \frac{2}{\pi d}\sum_{i=1}^d \beta_i \frac{\sin\theta_i}{1 + \sin^2\theta_i} = 0
\end{equation*}
which can be re-written as
\begin{equation}\label{eq: pde_to_be_solved}
    \sum_{i=1}^d \alpha_i \frac{\partial}{\partial \theta_i}f_0(\thetab) = \frac{2}{\pi d}
    f_0(\thetab) \sum_{i=1}^d \beta_i \frac{\sin\theta_i}{1 + \sin^2\theta_i}.
\end{equation}
Our goal is to solve the partial differential equation (PDE) \eqref{eq: pde_to_be_solved} with the initial condition $\thetab=\thetab_0$, using the method of characteristics (for more details see \cite{pdes_book}). 
Note that if $\alpha_1 = \ldots = \alpha_d = 0$, then \eqref{eq: pde_to_be_solved} cannot hold for all $\thetab \in \mathcal{T}_d$, so we require at least one $\alpha_i \neq 0.$
The transport structure of this PDE leads to a family of curves, known as the characteristic curves, along which the PDE \eqref{eq: pde_to_be_solved} reduces to an ODE. 
Those curves are classically defined by integrating the transport coefficients, and are here given by
\begin{align*}
\frac{\partial}{\partial t} \varphi(t,\thetab_{0}) = \alphab.
\end{align*}
Combining this with the initial condition $\varphi(0,\thetab_0)=\thetab_0$, the characteristic curves are the straight lines given by
\begin{align*}
    \varphi(t,\thetab_{0}) & = \thetab_{0} + t\alphab.
\end{align*}

\noindent Next, in order to simplify the PDE \eqref{eq: pde_to_be_solved} by absorbing its right-hand side, we define 
\begin{align*}
   h_0(\thetab)    & = f_0(\thetab) \exp \left( - \frac{2}{\pi d}\frac{1}{2 \sqrt{2}} \sum^d_{i=1} \frac{\beta_i}{\alpha_i} \log \left( \frac{\sec ^2\left(\frac{\theta_i}{2}\right)-2 \sqrt{2}+2}{\sec ^2\left(\frac{\theta_i}{2}\right)+2 \sqrt{2}+2} \right) \right) \\
   & = f_0(\thetab) \exp \left( \log \left( \prod^d_{i=1} \left( \frac{\sec ^2\left(\frac{\theta_i}{2}\right)-2 \sqrt{2}+2}{\sec ^2\left(\frac{\theta_i}{2}\right)+2 \sqrt{2}+2} \right)^{-\frac{\beta_i}{\sqrt{2}\pi d\alpha_i}} \right) \right) \\
   & = f_0(\thetab) \prod^d_{i=1} \left( \frac{\sec ^2\left(\frac{\theta_i}{2}\right)-2 \sqrt{2}+2}{\sec ^2\left(\frac{\theta_i}{2}\right)+2 \sqrt{2}+2} \right) ^{-\frac{\beta_i}{\sqrt{2}\pi d\alpha_i}} 
   = f_0(\thetab) \prod^d_{i=1}  \left( \frac{\sec ^2\left(\frac{\theta_i}{2}\right)-2 \sqrt{2}+2}{\sec ^2\left(\frac{\theta_i}{2}\right)+2 \sqrt{2}+2} \right)^{\delta_i}
\end{align*}

\noindent where 
\begin{equation*}
    \delta_i = \left\{ \begin{array}{cc}
       -\frac{\beta_i}{\sqrt{2}\pi d\alpha_i} & \alpha_i \neq 0 \\
        -\frac{\beta_i}{\sqrt{2}\pi d} & \alpha_i = 0
    \end{array} \right. .
\end{equation*}
We note that
\begin{equation}\label{eq: integral of weird sines}
    \frac{d}{d \theta_i} \log \left( \frac{\sec ^2\left(\frac{\theta_i}{2}\right)-2 \sqrt{2}+2}{\sec ^2\left(\frac{\theta_i}{2}\right)+2 \sqrt{2}+2} \right)
    = \frac{2\sqrt{2} \sin\theta_i}{1 + \sin^2\theta_i}.
\end{equation}
Solving \eqref{eq: pde_to_be_solved} is equivalent to solving the homogeneous PDE
\begin{align}
    \sum^d_{i=1} \alpha_i \frac{\partial h_0(\thetab)}{\partial \theta_i}
    = & \sum^d_{i=1} \alpha_i \frac{\partial f_0(\thetab)}{\partial \theta_i} \exp \left(- \frac{1}{\sqrt{2}\pi d} \sum^d_{i=1} \frac{\beta_i}{\alpha_i} \log \left( \frac{\sec ^2\left(\frac{\theta_i}{2}\right)-2 \sqrt{2}+2}{\sec ^2\left(\frac{\theta_i}{2}\right)+2 \sqrt{2}+2} \right) \right) \nonumber \\
    & - f_0(\thetab) \sum_{i=1}^d \frac{2\beta_i}{\pi d} \frac{\sin\theta_i}{1 + \sin^2\theta_i} \exp \left(- \frac{1}{\sqrt{2}\pi d} \sum^d_{i=1} \frac{\beta_i}{\alpha_i} \log \left( \frac{\sec ^2\left(\frac{\theta_i}{2}\right)-2 \sqrt{2}+2}{\sec ^2\left(\frac{\theta_i}{2}\right)+2 \sqrt{2}+2} \right) \right) 
    =  \ 0.\label{eq: pde_to_be_solved h}
\end{align}
We then derive the simple ODE verified by the solution of this transport problem along the characteristic curves:
\begin{equation*}
     \frac{\partial h_0(\varphi(t,\thetab_0))}{\partial t} = \sum^d_{i=1}\frac{\partial h_0}{\partial\theta_i}(\varphi(t,\thetab_0))\cdot \frac{\partial (\thetab_{0} +t \alphab)}{\partial t}= \sum^d_{i=1} \alpha_i \frac{\partial h_0}{\partial \theta_i}(\varphi(t,\thetab_0)).
\end{equation*}
Reminding its equivalent formulation \eqref{eq: pde_to_be_solved h}, the PDE \eqref{eq: pde_to_be_solved} is true if and only if for all $t,\thetab_0$, 
$$\frac{\partial h_0(\varphi(t,\thetab_0))}{\partial t}=0,$$
namely if $h_0$ is constant along each characteristic curve. This holds if and only if there exist $\alpha_i \neq 0$ 
such that, for all $t\in \mathbb{R}$ and $\thetab_0 \in \mathcal{T}_d$, we have
    $h_0(\thetab_0 + t\alphab) = h_0(\thetab_0).$
\end{proof}

\begin{proof}[\rm\textbf{Proof of Theorem~\ref{theo:margins}.}]
Under the condition that $f_0(\thetab_m, \thetab_{m^c}) = f_0(\thetab_m, -\thetab_{m^c})$, it holds that
\begin{align*}
   f_m(\thetab_m) 
   & = \int_{\mathcal{T}_d} f_0 (\thetab) \left( 1 + \frac{2}{\pi d} \sum_{i=1}^d \arctan \left( \frac{\lambda_i \sin (\theta_i)}{1 + \sin^2 (\theta_i)} \right) \right) \mathrm{d}\theta_{m^c} \\
   & = f_{0;m} (\thetab_m) \left( 1 + \frac{2}{\pi d} \sum_{i\in m} \arctan \left( \frac{\lambda_i \sin (\theta_i)}{1 + \sin^2 (\theta_i)} \right) \right)  + \frac{2}{\pi d} \sum_{i \in m^c} \int_{-\pi}^\pi \cdots \int_{-\pi}^\pi f_0 (\thetab) \arctan \left( \frac{\lambda_i \sin (\theta_i)}{1 + \sin^2 (\theta_i)} \right) \mathrm{d}\theta_{m^c} \\
   & = f_{0;m} (\thetab_m) \left( 1 + \frac{2}{\pi d} \sum_{i\in m} \arctan \left( \frac{\lambda_i \sin (\theta_i)}{1 + \sin^2 (\theta_i)} \right) \right)  + \frac{2}{\pi d} \sum_{i \in m^c} \int_{0}^\pi \cdots \int_{0}^\pi f_0 (\thetab) \arctan \left( \frac{\lambda_i \sin (\theta_i)}{1 + \sin^2 (\theta_i)} \right) \mathrm{d}\theta_{m^c} \\
   & + \frac{2}{\pi d} \sum_{i \in m^c} \int_{0}^\pi \cdots \int_{0}^\pi f_0 (\thetab_m, -\thetab_{m^c}) \arctan \left( \frac{\lambda_i \sin (-\theta_i)}{1 + \sin^2 (-\theta_i)} \right) \mathrm{d}\theta_{m^c}\\
   & = f_{0;m} (\thetab_m) \left( 1 + \frac{2}{\pi d} \sum_{i\in m} \arctan \left( \frac{\lambda_i \sin (\theta_i)}{1 + \sin^2 (\theta_i)} \right) \right).
\end{align*}
\end{proof}

\begin{proof}[\rm\textbf{Proof of Theorem~\ref{theo:conds}.}]
Straightforward calculations yield
\begin{align*}
    f_{m|m^c}(\thetab_m| \thetab_{m^c}) 
    & = \frac{f(\thetab)}{f_{m^c}(\thetab_{m^c})}
    = \frac{f_0 (\thetab) \left( 1 + \frac{2}{\pi d} \sum_{i=1}^d \arctan \left( \frac{\lambda_i \sin (\theta_i)}{1 + \sin^2 (\theta_i)} \right) \right)}{f_{0;m^c} (\thetab_{m^c}) \left( 1 + \frac{2}{\pi d} \sum_{i\in m^c} \arctan \left( \frac{\lambda_i \sin (\theta_i)}{1 + \sin^2 (\theta_i)} \right) \right)} \\
    & = f_{0;m|m^c} (\thetab_{m} | \thetab_{m^c}) \frac{1 + \frac{2}{\pi d} \sum_{i \in m} \arctan \left( \frac{\lambda_i \sin (\theta_i)}{1 + \sin^2 (\theta_i)} \right) + \frac{2}{\pi d} \sum_{i \in m^c} \arctan \left( \frac{\lambda_i \sin (\theta_i)}{1 + \sin^2 (\theta_i)} \right)}{1 + \frac{2}{\pi d} \sum_{i\in m^c} \arctan \left( \frac{\lambda_i \sin (\theta_i)}{1 + \sin^2 (\theta_i)} \right)} \\
    & = f_{0;m|m^c} (\thetab_{m} | \thetab_{m^c}) \left(1 + \frac{\frac{2}{\pi d} \sum_{i \in m} \arctan \left( \frac{\lambda_i \sin (\theta_i)}{1 + \sin^2 (\theta_i)} \right)}{1 + \frac{2}{\pi d} \sum_{i\in m^c} \arctan \left( \frac{\lambda_i \sin (\theta_i)}{1 + \sin^2 (\theta_i)} \right)} \right).
\end{align*}
\end{proof}

\begin{proof}[\rm\textbf{Proof of Theorem~\ref{Theor: moments}.}]
Simple manipulations yield
\begin{align*}
    \alpha_{\pb} 
    & = \E[\cos( \pb^\intercal \Thetab)]
    = \int_{\mathcal{T}_d} \cos(\pb^\intercal \thetab) f(\thetab) \mathrm{d} \thetab \\
    & = \int_{\mathcal{T}_d} \cos(\pb^\intercal \thetab) f_0 (\thetab) \left( 1 + \frac{2}{\pi d} \sum_{i=1}^d \arctan \left( \frac{\lambda_i \sin (\theta_i)}{1 + \sin^2 (\theta_i)} \right) \right) \mathrm{d} \thetab \\
    & = \int_{\mathcal{T}_d} \cos(\pb^\intercal \thetab) f_0 (\thetab) \mathrm{d} \thetab + \frac{2}{\pi d} \sum_{i=1}^d \int_{\mathcal{T}_d} \cos(\pb^\intercal \thetab) f_0 (\thetab)  \arctan \left( \frac{\lambda_i \sin (\theta_i)}{1 + \sin^2 (\theta_i)} \right) \mathrm{d} \thetab  = \alpha^0_{\pb}
\end{align*}
where the last equality follows from integrating an odd function over a symmetric domain. With a similar reasoning we obtain 
\begin{align}\label{eq: cosine_moments}
    \beta_{\pb} 
    & = \E[\sin( \pb^\intercal \Thetab)]
    = \int_{\mathcal{T}_d} \sin(\pb^\intercal \thetab) f(\thetab) \mathrm{d} \thetab \nonumber \\
    & = \int_{\mathcal{T}_d} \sin(\pb^\intercal \thetab) f_0 (\thetab) \left( 1 + \frac{2}{\pi d} \sum_{i=1}^d \arctan \left( \frac{\lambda_i \sin (\theta_i)}{1 + \sin^2 (\theta_i)} \right) \right) \mathrm{d} \thetab \nonumber \\
    & = \int_{\mathcal{T}_d} \sin(\pb^\intercal \thetab) f_0 (\thetab) \mathrm{d} \thetab + \frac{2}{\pi d} \sum_{i=1}^d \int_{\mathcal{T}_d} \sin(\pb^\intercal \thetab) f_0 (\thetab)  \arctan \left( \frac{\lambda_i \sin (\theta_i)}{1 + \sin^2 (\theta_i)} \right) \mathrm{d} \thetab \nonumber \\
    & = \frac{2}{\pi d} \sum_{i=1}^d \int_{\mathcal{T}_d} \sin(\pb^\intercal \thetab) f_0 (\thetab)  \arctan \left( \frac{\lambda_i \sin (\theta_i)}{1 + \sin^2 (\theta_i)} \right) \mathrm{d} \thetab.
\end{align}
Using Fourier {sine series}, we can rewrite 
\begin{equation*}
    \arctan \left( \frac{\lambda_i \sin (\theta_i)}{1 + \sin^2 (\theta_i)} \right)
    = \sum_{n=1}^\infty b_n (\lambda_i) \sin \left( n \theta_i \right) 
\end{equation*}
for $b_n (\lambda_i)$ as in~\eqref{Fourcoeff} since $\arctan \left( \frac{\lambda_i \sin (\theta_i)}{1 + \sin^2 (\theta_i)} \right)$ is {odd and} periodic on $\mathcal{T}_1$. With this in hand~\eqref{eq: cosine_moments} turns into 
\begin{align*}
    \beta_{\pb} 
    & = \frac{2}{\pi d} \sum_{i=1}^d \int_{\mathcal{T}_d} \sin(\pb^\intercal \thetab) f_0 (\thetab)  \sum_{n=1}^\infty b_n (\lambda_i) \sin \left( n \theta_i \right)  \mathrm{d} \thetab \\
    & = \frac{2}{\pi d} \sum_{i=1}^d \sum_{n=1}^\infty  b_n (\lambda_i) \int_{\mathcal{T}_d} \sin(\pb^\intercal \thetab) f_0 (\thetab) \sin \left( n \theta_i \right)  \mathrm{d} \thetab \\
    & = \frac{2}{\pi d} \sum_{i=1}^d \sum_{n=1}^\infty  b_n (\lambda_i) \int_{\mathcal{T}_d} \frac{1}{2} \Biggl( \cos(p_1 \theta_1 + \cdots + p_{i-1} \theta_{i-1} + (p_i - n) \theta_i + p_{i+1} \theta_{i+1} + \cdots + p_d \theta_d) \\
    & \hspace{2cm} - \cos(p_1 \theta_1 + \cdots + p_{i-1} \theta_{i-1} + (p_i + n) \theta_i + p_{i+1} \theta_{i+1} + \cdots + p_d \theta_d) \Biggr) f_0 (\thetab)  \mathrm{d} \thetab \\
    & = \frac{1}{\pi d} \sum_{i=1}^d \sum_{n=1}^\infty  b_n (\lambda_i) \left( \alpha^0_{[p_1, \ldots, p_{i-1}, p_i-n, p_{i+1}, \ldots, p_d]} - \alpha^0_{[p_1, \ldots, p_{i-1}, p_i+n, p_{i+1}, \ldots, p_d]} \right) \\
    & = \frac{1}{\pi d} \sum_{i=1}^d \sum_{n=1}^\infty  b_n (\lambda_i) \left( \alpha^0_{\pb - n \eb_i} - \alpha^0_{\pb + n \eb_i} \right),
\end{align*}
where the second equality follows from the boundedness of all involved functions.
\end{proof}

\begin{proof}[\rm\textbf{Proof of Theorem~\ref{thm: random_variate_generation}.}]
    For notational convenience, we define
    \begin{align*}
        A_{\Xb} \coloneqq \left\{ U \leq \frac{1}{2} \left( 1 + \frac{2}{\pi d} \sum_{i=1}^d \arctan \left( \frac{\lambda_i \sin (X_i - \mu_i)}{1 + \sin^2 (X_i - \mu_i)} \right) \right) \right\}
    \end{align*}
    and note that 
    \begin{align*}
        A_{2\mub - \Xb} = \left\{ U \leq \frac{1}{2} \left( 1 + \frac{2}{\pi d} \sum_{i=1}^d \arctan \left( \frac{\lambda_i \sin (\mu_i - X_i)}{1 + \sin^2 (\mu_i - X_i)} \right) \right) \right\}
        = \left\{ U \leq \frac{1}{2} \left( 1 - \frac{2}{\pi d} \sum_{i=1}^d \arctan \left( \frac{\lambda_i \sin (X_i - \mu_i)}{1 + \sin^2 (X_i - \mu_i)} \right) \right) \right\}.
    \end{align*}
    For $\thetab \in \mathcal{T}_d$, it holds that
    \begin{align*}
        \Prob \left( \Thetab = \thetab \right) 
        = & \Prob \left( \Xb = \thetab, A_{\Xb} \right) + \Prob \left( \Xb = 2\mub - \thetab, A^c_{2\mub - \Xb} \right) \\
        = & \Prob \left( \Xb = \thetab \right) \Prob \left( A_{\Xb} \right) + \Prob \left( \Xb = 2\mub - \thetab \right) \Prob \left( A^c_{2\mub - \Xb} \right) \\
        = & f_0(\thetab - \mub) \frac{1}{2} \left( 1 + \frac{2}{\pi d} \sum_{i=1}^d \arctan \left( \frac{\lambda_i \sin (X_i - \mu_i)}{1 + \sin^2 (X_i - \mu_i)} \right) \right) \\
        & + f_0(\mub - \thetab) \left[ 1 - \frac{1}{2} \left( 1 - \frac{2}{\pi d} \sum_{i=1}^d \arctan \left( \frac{\lambda_i \sin (X_i - \mu_i)}{1 + \sin^2 (X_i - \mu_i)} \right) \right) \right] \\
        = & f_0(\thetab - \mub) \left( 1 + \frac{2}{\pi d} \sum_{i=1}^d \arctan \left( \frac{\lambda_i \sin (X_i - \mu_i)}{1 + \sin^2 (X_i - \mu_i)} \right) \right),
    \end{align*}
    where we used that $f_0 (\thetab - \mub) = f_0(\mub - \thetab)$ thanks to the symmetry around $\mub$.
\end{proof}

\begin{proof}[\rm{\textbf{Elements of the proof of Proposition~\ref{prop: estimate quantities}.}}]
Parts (ii), (iii) and (iv) can be proved the same way as Propositions 4.1 and 4.2 of \cite{symmetry_test}.
For the proof of (i), we note that 

\begin{equation*}
    \frac{\sec ^2 \left( \left( \theta_{ji} - \hat{\mu}^{(n)}_j\right)/{2} \right) -2 \sqrt{2}+2}{\sec ^2 \left( \left(\theta_{ji} - \hat{\mu}^{(n)}_j \right)/{2} \right)+2 \sqrt{2}+2} 
    = 1 - \frac{4 \sqrt{2}}{\sec ^2 \left( \left(\theta_{ji} - \hat{\mu}^{(n)}_j \right)/{2} \right)+2 \sqrt{2}+2} 
    \in \left[ \frac{3 - 2\sqrt{2}}{3 + 2\sqrt{2}}, 1 \right]
\end{equation*}
so
\begin{equation} \label{eq: log bound}
    \left| \log \left( \frac{\sec ^2 \left( \left( \theta_{ji} - \hat{\mu}^{(n)}_j\right)/{2} \right) -2 \sqrt{2}+2}{\sec ^2 \left( \left(\theta_{ji} - \hat{\mu}^{(n)}_j \right)/{2} \right)+2 \sqrt{2}+2} \right) \right| 
    \leq \log \left( \frac{3 + 2\sqrt{2}}{3 - 2\sqrt{2}} \right).
\end{equation}
Using \eqref{eq: log bound} instead of the inequality bounding the difference of two cosines by 2 in the proof of Proposition 4.1 of \cite{symmetry_test}, (i) can be proven along the same lines as in \cite{symmetry_test}.
    
\end{proof}

\section*{Supplementary Material}
The Supplementary Material contains the observed Fisher Information matrix and some simulations related to the tests of Section~\ref{sec: optimtest}.

\section*{Code availability}
The relevant code of the paper is available in \url{https://github.com/Sophia-Loizidou/Arctan-skewing}.

\section*{Acknowledgments}
Sophia Loizidou was supported by the grant PRIDE/21/16747448/MATHCODA from the Luxembourg National Research Fund. Both authors thank Toshihiro Abe and Chiara Passamonti for helpful comments on an early draft of the paper and Thomas Hamelryck and Ola R\o nning for the protein data.

\bibliographystyle{abbrv}
\bibliography{references}

\clearpage

\begingroup
\renewcommand{\thepage}{S\arabic{page}} 
\renewcommand{\thesection}{S\arabic{section}} 
\renewcommand{\thefigure}{S\arabic{figure}}
\renewcommand{\thetable}{S\arabic{table}}
\renewcommand{\theequation}{S\arabic{equation}}
\renewcommand{\thetheorem}{S\arabic{theorem}}
\renewcommand{\thelemma}{S\arabic{lemma}}
\renewcommand{\theassumption}{S\arabic{assumption}}

\setcounter{page}{1}                    
\setcounter{section}{0}
\setcounter{equation}{0}
\setcounter{figure}{0}
\setcounter{table}{0}

\begin{center}
  \LARGE \textbf{Supplement to ``Arctan-skewing: a new skewing scheme on the hypertorus''} 
\end{center}

\section{Elements of the observed Fisher Information matrix}\label{sec: FIM}

In this section we provide the elements of the observed Fisher Information matrix for our arctan-skewed models, first in the general case and then what they become when $\lambdab=\zerob$. Before that, we give the following results:
\begin{align*}
    \frac{\partial}{\partial_{\mu_k}}M(\mu_k,\lambda_k; \theta_{jk})
    = & -\frac{2\lambda_k}{\pi d} \frac{\cos^2 (\theta_{jk} - \mu_k) \sin (\theta_{jk} - \mu_k)}{\left( \left(1 + \sin^2 (\theta_{jk} - \mu_k) \right)^2+\lambda_k^2 \sin^2 (\theta_{jk} - \mu_k) \right)^2}\\
    &\times \left\{ {2 \cos^2 (\theta_{jk} - \mu_k)  \left( 2 \left( 1+\sin^2 (\theta_{jk} - \mu_k) \right) + \lambda_k^2 \right)}  + {3 \left( \left( 1 + \sin^2 (\theta_{jk} - \mu_k) \right)^2+\lambda_k^2 \sin^2 (\theta_{jk} - \mu_k) \right)}\right\},
\end{align*}

$$
\frac{\partial}{\partial{\lambda_k}} M(\mu_k,\lambda_k; \theta_{jk})
= -\frac{2\cos^3 (\theta_{jk} - \mu_k)}{\pi d}\frac{\left( 1 + \sin^2 (\theta_{jk} - \mu_k) \right)^2-\lambda_k^2 \sin^2 (\theta_{jk} - \mu_k)}{\left((1 + \sin^2 (\theta_{jk} - \mu_k))^2+\lambda_k^2 \sin^2 (\theta_{jk} - \mu_k)\right)^2},
$$
and
$$
\frac{\partial}{\partial{\lambda_k}}L(\mu_k,\lambda_k; \theta_{jk})
= - \frac{4\lambda_k}{\pi d}\frac{ \sin^3 (\theta_{jk} - \mu_k) \left( 1 + \sin^2 (\theta_{jk} - \mu_k) \right)}
{\left( \left( 1 +  \sin^2 (\theta_{jk} - \mu_k) \right)^2 + \lambda_k^2 \sin^2 (\theta_{jk} - \mu_k)\right)^2}.
$$
Making use of the newly established notations and these results, we get
\begin{align*}
    \frac{\partial^2}{\partial^2{\mu_k}}\log\left(\prod_{j=1}^n f (\thetab_j - \mub) \right)
    = & \sum_{j=1}^n \frac{\left(\frac{\partial^2}{\partial^2 {\mu_k}} f_0(\thetab_j-\mub)\right)f_0(\thetab_j-\mub)-\left( \frac{\partial}{\partial{\mu_k}} f_0(\thetab_j-\mub) \right)^2}{f_0^2(\thetab_j-\mub)} \\
    & + \sum_{j=1}^n\frac{\frac{\partial}{\partial{\mu_k}}M(\mu_k, \lambda_k; \theta_{jk})}{ 1 + \frac{2}{\pi d} \sum_{i=1}^d \arctan \left( \frac{\lambda_i \sin (\theta_{ji} - \mu_{ji})}{1 + \sin^2 (\theta_{ji} - \mu_{ji})} \right)}\\
    & - \sum_{j=1}^n\frac{\left( M(\mu_k, \lambda_k; \theta_{jk}) \right)^2}{\left(1 + \frac{2}{\pi d} \sum_{i=1}^d \arctan \left( \frac{\lambda_i \sin (\theta_{ji} - \mu_{ji})}{1 + \sin^2 (\theta_{ji} - \mu_{ji})} \right)\right)^2 },
\end{align*}

\begin{align*}
\frac{\partial^2}{\partial{\mu_\ell}\partial{\mu_k}}\log\left(\prod_{j=1}^n f (\thetab_j - \mub) \right)
= & \sum_{j=1}^n \frac{\left(\frac{\partial^2}{\partial{\mu_{\ell}}\partial{\mu_k}}f_0(\thetab_j-\mub)\right) f_0(\thetab_j-\mub)-\frac{\partial}{\partial{\mu_{\ell}}}f_0(\thetab_j-\mub) \frac{\partial}{\partial{\mu_k}}f_0(\thetab_j-\mub) }{f_0^2(\thetab_j-\mub)}\\
&- \sum_{j=1}^n \frac{M(\mu_k,\lambda_k; \theta_{jk})M(\mu_\ell,\lambda_\ell; \theta_{j\ell})}
{ \left(1 + \frac{2}{\pi d} \sum_{i=1}^d \arctan \left( \frac{\lambda_i \sin (\theta_{ji} - \mu_{i})}{1 + \sin^2 (\theta_{ji} - \mu_{i})} \right)\right)^2 },
\end{align*}

\begin{align*}
\frac{\partial^2}{\partial{\lambda_k}\partial{\mu_k}}\log\left(\prod_{j=1}^n f (\thetab_j - \mub) \right)
= & \sum_{j=1}^n\frac{\frac{\partial}{\partial{\lambda_k}}M(\mu_k,\lambda_k; \theta_{jk})}
{ 1 + \frac{2}{\pi d} \sum_{i=1}^d \arctan \left( \frac{\lambda_i \sin (\theta_{ji} - \mu_{ji})}{1 + \sin^2 (\theta_{ji} - \mu_{ji})} \right) } \\
& - \sum_{j=1}^n\frac{M(\mu_k,\lambda_k; \theta_{jk})L(\mu_k,\lambda_k; \theta_{jk})}
{ \left(1 + \frac{2}{\pi d} \sum_{i=1}^d \arctan \left( \frac{\lambda_i \sin (\theta_{ji} - \mu_{ji})}{1 + \sin^2 (\theta_{ji} - \mu_{ji})} \right)\right)^2 },
\end{align*}

\begin{eqnarray*}
\frac{\partial^2}{\partial{\lambda_\ell}\partial{\mu_k}}\log\left(\prod_{j=1}^n f (\thetab_j - \mub) \right)&=&-\sum_{j=1}^n\frac{M(\mu_k,\lambda_k; \theta_{jk})L(\mu_\ell,\lambda_\ell; \theta_{j\ell})}
{\left(1 + \frac{2}{\pi d} \sum_{i=1}^d \arctan \left( \frac{\lambda_i \sin (\theta_{ji} - \mu_{ji})}{1 + \sin^2 (\theta_{ji} - \mu_{ji})} \right) \right)^2 },
\end{eqnarray*}

\begin{eqnarray*}
    \frac{\partial^2}{\partial{\lambda_\ell}\partial{\lambda_k}}\log\left(\prod_{j=1}^n f (\thetab_j - \mub) \right)&=&-\sum_{j=1}^n \frac{L(\mu_k,\lambda_k; \theta_{jk})L(\mu_\ell,\lambda_\ell; \theta_{jk})}
{\left(1 + \frac{2}{\pi d} \sum_{i=1}^d \arctan \left( \frac{\lambda_i \sin (\theta_{ji} - \mu_{ji})}{1 + \sin^2 (\theta_{ji} - \mu_{ji})} \right)\right)^2},
\end{eqnarray*}
and finally
\begin{align*}
    \frac{\partial^2}{\partial{\lambda_k}\partial{\lambda_k}}\log\left(\prod_{j=1}^n f (\thetab_j - \mub) \right)
    = & \sum_{j=1}^n\frac{\frac{\partial}{\partial{\lambda_k}}L(\mu_k,\lambda_k; \theta_{jk})} {1 + \frac{2}{\pi d} \sum_{i=1}^d \arctan \left( \frac{\lambda_i \sin (\theta_{ji} - \mu_{ji})}{1 + \sin^2 (\theta_{ji} - \mu_{ji})} \right)} \\
    & - \sum_{j=1}^n\frac{L^2(\mu_k,\lambda_k; \theta_{jk})}
    {\left(1 + \frac{2}{\pi d} \sum_{i=1}^d \arctan \left( \frac{\lambda_i \sin (\theta_{ji} - \mu_{ji})}{1 + \sin^2 (\theta_{ji} - \mu_{ji})} \right)\right)^2}.
\end{align*}

When $\lambdab=\zerob$, these expressions  simplify into
\begin{align*}
\left. \frac{\partial^2}{\partial^2 {\mu_k}} \log\left(\prod_{j=1}^n f (\thetab_j - \mub) \right) \right|_{\lambdab=\zerob}
& = \sum_{j=1}^n \frac{\left(\frac{\partial^2}{\partial^2 {\mu_k}} f_0(\thetab_j-\mub)\right)f_0(\thetab_j-\mub)-\left( \frac{\partial}{\partial{\mu_k}}f_0(\thetab_j-\mub) \right)^2}{f_0^2(\thetab_j-\mub)},\\
\left. \frac{\partial^2}{\partial {\mu_\ell} \partial {\mu_k}} \log\left(\prod_{j=1}^n f (\thetab_j - \mub) \right) \right|_{\lambdab=\zerob}
& = \sum_{j=1}^n \frac{ \left( \frac{\partial^2}{\partial {\mu_\ell} \partial {\mu_k}} f_0(\thetab_j-\mub)\right) f_0(\thetab_j-\mub)-\frac{\partial}{\partial{\mu_{\ell}}}f_0(\thetab_j-\mub) \frac{\partial}{\partial{\mu_k}}f_0(\thetab_j-\mub) }{f_0^2(\thetab_j-\mub)},\\
\left. \frac{\partial^2}{\partial {\lambda_k} \partial {\mu_k}} \log\left(\prod_{j=1}^n f (\thetab_j - \mub) \right) \right|_{\lambdab=\zerob}
& = -\sum_{j=1}^n \frac{2\cos^3 (\theta_{jk} - \mu_k)}{\pi d \left(1 + \sin^2 (\theta_{jk} - \mu_k) \right)^2}, \\
\left. \frac{\partial^2}{\partial {\lambda_\ell} \partial {\mu_k}} \log\left(\prod_{j=1}^n f (\thetab_j - \mub) \right) \right|_{\lambdab=\zerob}
& = 0, \\
\left. \frac{\partial^2}{\partial {\lambda_\ell} \partial {\lambda_k}} \log\left(\prod_{j=1}^n f (\thetab_j - \mub) \right) \right|_{\lambdab=\zerob}
& = \sum_{j=1}^n \frac{4}{\pi^2 d^2}\frac{ \sin (\theta_{jk} - \mu_k)}
{1 +  \sin^2 (\theta_{jk} - \mu_k)}\frac{ \sin (\theta_{j\ell} - \mu_\ell)}
{1 +  \sin^2 (\theta_{j\ell} - \mu_\ell)},
\end{align*}
and
$$
\left. \frac{\partial^2}{\partial^2 {\lambda_k}} \log\left(\prod_{j=1}^n f (\thetab_j - \mub) \right) \right|_{\lambdab=\zerob}
= - \sum_{j=1}^n\frac{4}{\pi^2 d^2}\left(\frac{ \sin (\theta_{jk} - \mu_k)}
{1 +  \sin^2 (\theta_{jk} - \mu_k)}\right)^2.
$$

\section{Simulations} \label{sec: simulations}

In this section we present the simulation results for the tests proposed in Section~\ref{sec: optimtest}, $\phi^{\ast (n); \mu}$ and $\phi^{\ast (n)}_{f_0}$, for the specified and unspecified symmetry center cases, respectively.
More specifically, we consider correctly and incorrectly specified symmetry center  $\mub$ in $\phi^{\ast (n); \mu}$ (denoted Specified and Misspecified, respectively), and for $\phi^{\ast (n)}_{f_0}$ (denoted Unspecified)  the assumed model $f_0$ is indicated in the tables. 
For each setting under consideration, we conduct a simulation study with 1000 replications.
The Monte Carlo estimates of the probability of rejection of the respective null hypotheses ($\mathcal{H}^{(n)}_{0; \mu}$ and $\mathcal{H}^{(n)}_{0}$) at the $5\%$ level of significance are presented in Tables~\ref{tab: arctan-sk}-~\ref{tab: misspecified}.
We focus on $d=3$ and use the trivariate wrapped Cauchy (TWC) copula \citep{TWCC, TWCM} with marginal distributions chosen to be the wrapped Cauchy distribution. 
This distribution is denoted by $TWC_{\rho_{12};\rho_{13};\rho_{23}}^ {\beta_1; \beta_2; \beta_3}$, where $\rho_{12}, \rho_{13}, \rho_{23}\in\R\setminus\{0\}$ are the copula parameters satisfying Equation (8) in \cite{TWCC} and $\beta_1, \beta_2, \beta_3>0$ are the parameters for the marginal distributions.

We firstly consider Table~\ref{tab: arctan-sk}, for which data are generated from the arctan-skewed TWC$^{0.5, 0.7, 0.9}_{0.25, 1, 4}$ model.
As expected, the power of the tests increases as the sample size and the skewness degree increase.
For the specified test, the desired percentage of rejection under the null hypothesis of symmetry is attained, while the unspecified test has a higher percentage of rejections.

The settings chosen Tables~\ref{tab: sine-sk} and~\ref{tab: misspecified}, described in the next paragraphs, both violate assumptions under which our optimal tests are built: the former because the skewing is not arctan, the latter because the center of symmetry is misspecified. 
It is relevant to know how our tests behave when different assumptions are not met.

In Table~\ref{tab: sine-sk}, data are generated from the sine-skewed TWC$^{0.5, 0.7, 0.9}_{0.25, 1, 4}$ model.
This table indicates that our tests can detect other types of skewing, although a larger sample size is required for deviations from symmetry to be detected consistently.

Finally, Table~\ref{tab: misspecified} considers the case when the symmetry center $\mub$ is slightly misspecified ($\mub = (0.05, 0.05, 0.05)^\intercal$ instead of $\mub=\zerob$) and data are generated from the arctan-skewed TWC$^{0.5, 0.7, 0.9}_{0.25, 1, 4}$ model.
It is easy to see that the percentage of rejections decreases as the skewness increases.
This can be explained by the fact that positively skewing data also slightly shifts their mean  to larger values, and so the skewness with respect to a different symmetry center can be hidden.

\begin{table}[ht]
\centering
\caption[Percentage of rejections of $\phi^{\ast (n); \mu}$ and $\phi^{\ast (n)}$.]{Percentage of rejections of Specified ($\phi^{\ast (n); \mu}$) and Unspecified ($\phi^{\ast (n)}$) tests when generating data from the arctan-skewed TWC$^{0.5, 0.7, 0.9}_{0.25, 1, 4}$ distribution.
The variable $n$ represents the sample size and $\lambdab$ the value of the skewness parameter. \label{tab: arctan-sk}}
\begin{tabular}{m{2.5em}|m{3em}m{3em}m{3em}|m{3em}m{3em}m{3em}|m{3em}m{3em}m{3em}}
    \hline
    \multicolumn{10}{l}{Data generated from arctan-skewed TWC$^{0.5, 0.7, 0.9}_{0.25, 1, 4}$} \\
    \hline
  Test & \multicolumn{3}{c|}{Specified} & \multicolumn{3}{c|}{Unspecified} & \multicolumn{3}{c}{Unspecified} \\
  & \multicolumn{3}{c|}{} & \multicolumn{3}{c|}{$f_0$ = TWC$^{0.5, 0.7, 0.9}_{0.25, 1, 4}$} & \multicolumn{3}{c}{$f_0$ = TWC$^{0.1, 0.1, 0.1}_{0.3, 1, 3}$} \\
   \hline
 \diagbox[height=8ex, width=3.7em]{\raisebox{0.5\height}{\enspace $n$}}{ \raisebox{-0.8\height}{$\lambdab$}} & $\begin{pmatrix}
     0 \\ 0 \\ 0 
 \end{pmatrix}$ & $\begin{pmatrix}
     1 \\ 0.5 \\ 1 
 \end{pmatrix}$ & $\begin{pmatrix}
     1 \\ 2 \\ 3 
 \end{pmatrix}$ & $\begin{pmatrix}
     0 \\ 0 \\ 0 
 \end{pmatrix}$ & $\begin{pmatrix}
     1 \\ 0.5 \\ 1 
 \end{pmatrix}$ & $\begin{pmatrix}
     1 \\ 2 \\ 3 
 \end{pmatrix}$ & $\begin{pmatrix}
     0 \\ 0 \\ 0 
 \end{pmatrix}$ & $\begin{pmatrix}
     1 \\ 0.5 \\ 1 
 \end{pmatrix}$ & $\begin{pmatrix}
     1 \\ 2 \\ 3 
 \end{pmatrix}$ \\ 
  \hline
250 & 0.054 & 0.178 & 0.568 & 0.114 & 0.277 & 0.647 & 0.111 & 0.271 & 0.624 \\ 
  500 & 0.054 & 0.382 & 0.886 & 0.096 & 0.436 & 0.897 & 0.071 & 0.436 & 0.894 \\ 
  750 & 0.045 & 0.538 & 0.976 & 0.073 & 0.547 & 0.977 & 0.069 & 0.558 & 0.969 \\ 
  1000 & 0.039 & 0.663 & 0.993 & 0.060 & 0.671 & 0.993 & 0.075 & 0.678 & 0.991 \\ 
\end{tabular}
\end{table}

\begin{table}[ht]
\centering
\caption[Percentage of rejections of $\phi^{\ast (n); \mu}$ and $\phi^{\ast (n)}$.]{Percentage of rejections of Specified ($\phi^{\ast (n); \mu}$) and Unspecified ($\phi^{\ast (n)}$) tests when generating data from the sine-skewed TWC$^{0.5, 0.7, 0.9}_{0.25, 1, 4}$ distribution.
The variable $n$ represents the sample size and $\lambdab$ the value of the skewness parameter. \label{tab: sine-sk}}
\begin{tabular}{m{2.5em}|m{3em}m{3em}m{3em}|m{3em}m{3em}m{3em}|m{3em}m{3em}m{3em}}
    \hline
    \multicolumn{10}{l}{Data generated from sine-skewed TWC$^{0.5, 0.7, 0.9}_{0.25, 1, 4}$} \\
    \hline
  Test & \multicolumn{3}{c|}{Specified} & \multicolumn{3}{c|}{Unspecified} & \multicolumn{3}{c}{Unspecified} \\
  & \multicolumn{3}{c|}{} & \multicolumn{3}{c|}{$f_0$ = TWC$^{0.5, 0.7, 0.9}_{0.25, 1, 4}$} & \multicolumn{3}{c}{$f_0$ = TWC$^{0.1, 0.1, 0.1}_{0.3, 1, 3}$} \\
   \hline
 \diagbox[height=8ex, width=3.7em]{\raisebox{0.5\height}{\enspace $n$}}{ \raisebox{-0.8\height}{$\lambdab$}} & $\begin{pmatrix}
     0 \\ 0 \\ 0 
 \end{pmatrix}$ & $\begin{pmatrix}
     0.2 \\ 0.2 \\ 0.2 
 \end{pmatrix}$ & $\begin{pmatrix}
     0.6 \\ 0.1 \\ 0.2 
 \end{pmatrix}$ & $\begin{pmatrix}
     0 \\ 0 \\ 0 
 \end{pmatrix}$ & $\begin{pmatrix}
     0.2 \\ 0.2 \\ 0.2 
 \end{pmatrix}$ & $\begin{pmatrix}
     0.6 \\ 0.1 \\ 0.2 
 \end{pmatrix}$ & $\begin{pmatrix}
     0 \\ 0 \\ 0 
 \end{pmatrix}$ & $\begin{pmatrix}
     0.2 \\ 0.2 \\ 0.2 
 \end{pmatrix}$ & $\begin{pmatrix}
     0.6 \\ 0.1 \\ 0.2 
 \end{pmatrix}$ \\ 
  \hline
500 & 0.049 & 0.086 & 0.280 & 0.075 & 0.139 & 0.292 & 0.081 & 0.144 & 0.284 \\ 
  1000 & 0.060 & 0.128 & 0.544 & 0.080 & 0.154 & 0.496 & 0.079 & 0.152 & 0.490 \\ 
  1500 & 0.058 & 0.181 & 0.722 & 0.063 & 0.192 & 0.653 & 0.067 & 0.202 & 0.634 \\ 
  2000 & 0.061 & 0.229 & 0.844 & 0.074 & 0.213 & 0.781 & 0.069 & 0.264 & 0.767 \\ 
\end{tabular}
\end{table}

\begin{table}[ht]
\centering
\caption[Percentage of rejections of $\phi^{\ast (n); \mu}$ and $\phi^{\ast (n)}$.]{Percentage of rejections of Specified and Misspecified ($\phi^{\ast (n); \mu}$) tests when generating data from the arctan-skewed TWC$^{0.5, 0.7, 0.9}_{0.25, 1, 4}$ distribution.
The variable $n$ represents the sample size and $\lambdab$ the value of the skewness parameter. \label{tab: misspecified}}
\begin{tabular}{m{2.5em}|m{3em}m{3em}m{3em}|m{3em}m{3em}m{3em}}
    \hline
    \multicolumn{7}{l}{Data generated from arctan-skewed TWC$^{0.5, 0.7, 0.9}_{0.25, 1, 4}$} \\
    \hline
  Test & \multicolumn{3}{c|}{Specified} & \multicolumn{3}{c}{Misspecified} \\
  & \multicolumn{3}{c|}{} & \multicolumn{3}{c}{$\mub = (0.05, 0.05, 0.05)^\intercal$} \\
   \hline
   \diagbox[height=8ex, width=3.7em]{\raisebox{0.5\height}{\enspace $n$}}{ \raisebox{-0.8\height}{$\lambdab$}} & $\begin{pmatrix}
     0 \\ 0 \\ 0 
 \end{pmatrix}$ & $\begin{pmatrix}
     1 \\ 0.5 \\ 1 
 \end{pmatrix}$ & $\begin{pmatrix}
     1 \\ 2 \\ 3 
 \end{pmatrix}$ & $\begin{pmatrix}
     0 \\ 0 \\ 0 
 \end{pmatrix}$ & $\begin{pmatrix}
     1 \\ 0.5 \\ 1 
 \end{pmatrix}$ & $\begin{pmatrix}
     1 \\ 2 \\ 3 
 \end{pmatrix}$ \\ 
  \hline
250 & 0.055 & 0.190 & 0.514 & 0.831 & 0.604 & 0.257 \\ 
  500 & 0.048 & 0.375 & 0.894 & 0.997 & 0.880 & 0.468 \\ 
  750 & 0.054 & 0.533 & 0.973 & 1.000 & 0.978 & 0.631 \\ 
  1000 & 0.050 & 0.667 & 0.998 & 1.000 & 0.999 & 0.808 \\
\end{tabular}
\end{table}

\end{document}